\documentclass[aps,prl,showpacs,notitlepage,twocolumn,superscriptaddress,nofootinbib,preprintnumbers,floatfix]{revtex4-2}

\usepackage{graphicx}% Include figure files
\usepackage{dcolumn}% Align table columns on decimal point
\usepackage{bm}% bold math
\usepackage[colorlinks=true,linkcolor=blue,citecolor=blue,urlcolor=blue]{hyperref}

\usepackage{bbold} % for identity 1
\usepackage{tcolorbox}
\usepackage{algorithm}
\usepackage{algpseudocode}
\usepackage{amsmath}
\usepackage{amsfonts}
\usepackage{amsthm}
\usepackage{amssymb}
\usepackage{mathtools}
\usepackage{braket}
\theoremstyle{plain}
\newtheorem{theorem}{Theorem}%[section]

\newtheorem{lemma}[theorem]{Lemma}

\theoremstyle{definition}

\newtheorem{assumption}[theorem]{Assumption}

\theoremstyle{remark}

\newtheorem*{theorem*}{Theorem}
\newtheorem*{lemma*}{Lemma}
\newtheorem*{definition*}{Definition}
\newtheorem*{corollary*}{Corollary}
\newtheorem*{remark*}{Remark}

\newcommand{\NN}{\mathbb{N}}
\newcommand{\cB}{\mathcal{B}}
\newcommand{\cH}{\mathcal{H}}
\newcommand{\cL}{\mathcal{L}}

\newcommand{\id}{\mathbb{1}}
\newcommand{\norm}[1]{\left\lVert #1 \right\rVert}

\begin{document}

%\title{Scaling Theory for Learning Low-Energy Quantum States \\ from Wavefunction Subspaces}
\title{Scaling Theory for Learning Low-Energy Quantum Subspaces}
\author{Guijing Duan}
\thanks{These authors contributed equally to this study.}
\affiliation{Department of Physics, Tsinghua University, Beijing 100084, China}

\author{Chen Mo}
\thanks{These authors contributed equally to this study.}
\affiliation{Department of Physics, Tsinghua University, Beijing 100084, China}

\author{Di Luo}
\email{diluo@tsinghua.edu.cn}
\affiliation{Department of Physics, Tsinghua University, Beijing 100084, China}
\affiliation{Institute of Advanced Study, Tsinghua University, Beijing 100084, China}

\begin{abstract}
Phase diagrams and energy surfaces require low-energy states across a family of Hamiltonians, but solving each parameter point independently is prohibitively expensive. An important question is not how a single labeled eigenstate behaves under parameter variation, but how much information about the low-energy sector is contained in a small set of reference wavefunctions. We show that the physically relevant object is the isolated low-energy subspace itself: it remains well defined through degeneracies and level crossings, and it can be learned efficiently from nearby sampled states. Working in this subspace, we prove that a sampling pattern cancelling the first $q$ nonconstant response orders yields an energy-density error bounded by $d^{2(q+1)}$ for every retained level, where $d$ is the parameter-space sampling distance. For a gapped local ground state on $n$ sites, locality sharpens this to $d^2(nd^2)^q$, identifying $nd^2$ as the natural scaling variable governing the effectiveness of higher-order information from sampled wavefunctions. In a local analytic regime this gives a sampling cost $K=\mathcal{O}([\log(1/\varepsilon)]^D)$ at fixed parameter dimension $D$. Numerical results on Heisenberg and transverse-field Ising chains confirm the distance scaling law and the $nd^2$ collapse. Beyond characterizing learnability, the scaling behavior further provides a finite-size probe of phase transitions near criticality.

\end{abstract}
\maketitle

\emph{Introduction.---}
Many problems in quantum many-body physics require solving not one Hamiltonian, but a family of closely related Hamiltonians.
Phase diagrams, excitation spectra, and molecular energy surfaces all require low-energy eigenstates as couplings, fields, or geometries vary~\cite{Ors2014,Haegeman2012,Tilly2022,McArdle2020}.
Solving each point independently is often the dominant cost.
The natural physical hope is that nearby Hamiltonians share enough low-energy structure that wavefunctions computed at a few representative points can predict the rest.
This raises a basic question: how much low-energy information about a Hamiltonian family is carried by a small set of reference wavefunctions?

This expectation underlies two rapidly developing directions.
Foundation neural quantum states aim to represent the ground states of many Hamiltonians with a single model, enabling transfer and fine-tuning across coupling regimes~\cite{Zhang2023,zaklama_arxiv2512.11962,Rende2025,Scherbela2022,gao_arxiv2110.05064,Fuarxiv.2603.02346}.
Reduced-basis methods, including eigenvector continuation, take a complementary route: they use already computed wavefunctions as a problem-adapted variational basis by projecting $H(\bm{x})$ onto their span~\cite{MejutoZaera2023,frame2018eigenvector, cances2002towards, Knig2020,  Yoshida2022,Ohlberger_arxiv1511.02021}.
The projected ground-state energy is no higher than the energy of any individual input state, independently of the solver used to generate those states.
The success of these methods across nuclear physics, quantum chemistry, and condensed matter suggests that low-energy wavefunctions often occupy an effectively low-dimensional manifold inside Hilbert space~\cite{Yapa2023,Christiansen2025,Bonilla2022,Anderson2022,Cheng2025, Prelovek2018,Franzke2024,Francis_arxiv2209.10571, Rath2025,Demol2020,Christiansen2025,Francis_arxiv.2209.10571,Agrawal2025}.
Recent foundational neural effective Hamiltonian constructions further combine the foundational neural quantum state viewpoint with this perspective, using neural quantum states sampled at selected couplings to accelerate scans of observables and phase boundaries~\cite{ZhangFNEH2026}. However, what has been missing from these advances is a scaling theory for learning the low-energy wavefunction subspace: how the prediction error depends on the sampling distance, the system size, and the number of sampled wavefunctions.

Such a theory must confront two many-body difficulties.
First, near degeneracies or level crossings, the identity of a labeled eigenstate can change abruptly.
Continuing a single eigenvector is then not a stable physical object, whereas the subspace spanned by all retained low-energy states remains well defined as long as it stays separated from higher levels.
Second, a local parameter change perturbs every site.
Thus a point that is close in parameter space can be far in Hilbert space, and the overlap between sampled and target ground states can degrade with the number of sites.
A useful controlled theory must therefore treat the low-energy sector as a whole and determine how accuracy depends simultaneously on sampling distance, the number of sampled points, and system size.

We develop such a scaling theory for learning low-energy quantum subspace from wavefunctions computed at nearby parameter points.
The construction combines sampled low-energy states into a variational subspace for the target Hamiltonian and treats the retained low-energy levels collectively, so it remains meaningful through degeneracies and level crossings.
The key mechanism is canceling response orders of the low-energy states under parameter variations.
A sampling pattern at parameter distance $d$ can cancel the first $q$ nonconstant response terms of the low-energy subspace.
The remaining leakage into high-energy states is then of order $d^{q+1}$, and because such leakage changes a variational energy only at second order, the energy-density error of every retained level is bounded by $d^{2(q+1)}$.
This is the scaling law for the isolated low-energy manifold: the exponent is fixed by the cancellation order and is unaffected by degeneracies or crossings inside the retained sector, while the prefactor contains the system-size dependence.
For a unique gapped ground state of a local system with $n$ sites, locality further sharpens this prefactor and gives the form $d^2(nd^2)^q$.
The combination $nd^2$ is then the physical small parameter: in a gapped local phase the fidelity susceptibility is extensive, so the infidelity between two ground states separated by distance $d$ is of order $nd^2$.
In a local analytic regime this scaling gives a sampled cost $K=\mathcal{O}([\log(1/\varepsilon)]^D)$ at fixed parameter dimension $D$.
Numerical results on gapped spin chains confirm the distance law and the $nd^2$ data collapse.
Finally, near a quantum critical point, where the relevant gap closes, the distance exponent departs from its gapped value; this departure provides a finite-size probe of the transition.

\emph{Low-energy subspace scaling law.---}
We first formulate the wavefunction-basis construction.
Consider a finite many-body system of size $n$ described by a Hamiltonian $H(\bm{x})$ that depends smoothly on $D$-dimension control parameters $\bm{x}=(x_1,\ldots,x_D)$.
We order its energy levels as $E_0(\bm{x})\le E_1(\bm{x})\le\cdots$ and focus on the subspace spanned by the lowest $M$ levels.
The only spectral assumption is an external gap: throughout a neighborhood $\|\bm{x}-\bm{x}_0\|\le d_0$ of the target point $\bm{x}_0$,
\begin{equation}
  E_M(\bm{x})-E_{M-1}(\bm{x})
  \ge \Delta_{\mathrm{high}}>0.
  \label{eq:external-gap}
\end{equation}
Crossings and degeneracies among $E_0(\bm{x}),\ldots,E_{M-1}(\bm{x})$ are allowed.
Individual eigenvectors inside this retained sector may rotate, but the whole low-energy subspace varies smoothly with $\bm{x}$~\cite{kato1966perturbation,Kato1995}.
We denote this low-energy subspace as $\mathcal{H}_{\bm{x}}$-subspace.

We sample the $\mathcal{H}$-subspace at neighboring points $\bm{x}_j=\bm{x}_0+d\bm{s}_j$, $j=1,\ldots,K$.
The dimensionless vectors $\bm{s}_j$ $(||\bm{s}||\geq1)$ define the sampling pattern and $d$ sets its overall radius, as illustrated in Fig.~\ref{fig:d-scaling}(a).
At each sampled parameter we choose any orthonormal basis $\{\ket{\psi_a(\bm{x}_j)}\}_{a=0}^{M-1}$ from the $\mathcal{H}_{\bm{x}}$-subspace and form the variational subspace
\begin{equation}
  \mathcal{B}_{M,K}(d)
  =
  \operatorname{span}
  \left\{
    \ket{\psi_a(\bm{x}_j)}
  \right\}_{\substack{
    a=0,\ldots,M-1\\
    j=1,\ldots,K
  }} .
  \label{eq:low-energy-reference-space}
\end{equation}
Because the $\mathcal{H}_{\bm{x}}$-subspace is included at each sampled parameter,  $\mathcal{B}_{M,K}(d)$ is unchanged by arbitrary basis rotations within this subspace.
The method therefore depends on the physical low-energy sector, not on a convention for labeling individual eigenstates. 

$\mathcal{B}_{M,K}(d)$ encodes the local response of the target subspace to parameter changes.
For $q\leq K$, $\mathcal{B}_{M,K}(d)$ contains the $q$-order response information of the target $\mathcal{H}_{\bm{x_0}}$-subspace because $\mathcal{B}_{M,K}(d)$ is able to cancel the first $q$ nonconstant response orders.
For example, $q=0$ uses the zeroth-order overlap with a nearby state, $q=1$ cancels the linear response, and $q=2$ also cancels the quadratic response.
After this cancellation, an $M$-dimensional subspace inside $\mathcal{B}_{M,K}(d)$ approximates the target $\mathcal{H}_{\bm{x_0}}$-subspace with error $\mathcal{O}(d^{q+1})$.
The explicit moment conditions on the sampling weights are given in the Supplemental Material~\cite{SM}.
In the following, we use min--max principle to turn this subspace accuracy into energy bounds.
\begin{figure}[t]
  \includegraphics[width=\columnwidth]{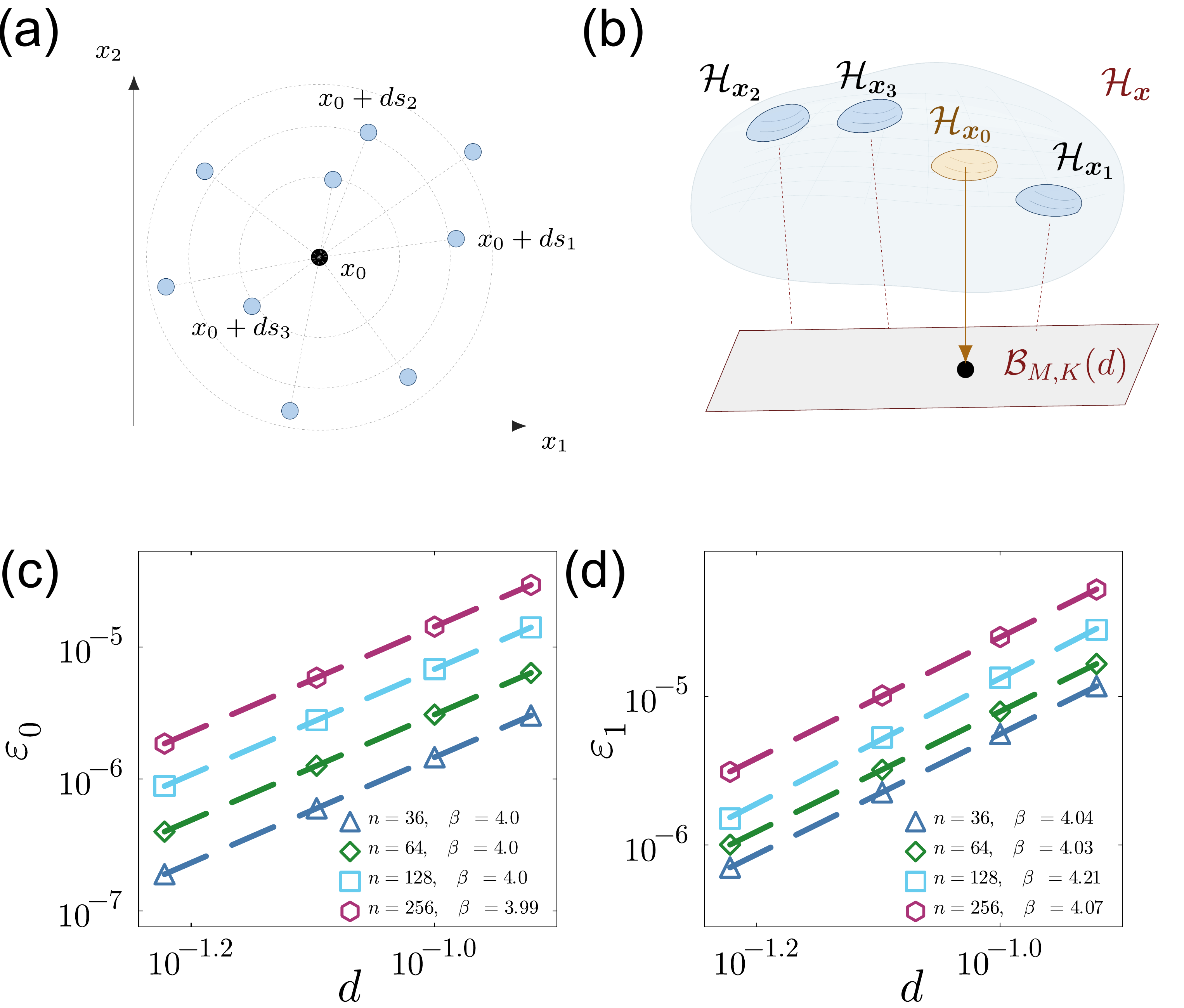}
  \caption{
    Geometry of the wavefunction basis construction and its distance scaling.
    (a) Sample $\mathcal{H}_{\bm{x_j}}$-subspaces at sampling
    distance $d$ around the target $\bm{x}_0$, which is itself not sampled.
    (b) The sampled $\mathcal{H}_{\bm{x_j}}$-subspaces  span 
    $\mathcal{B}_{M,K}(d)$, which approximates the target $\mathcal{H}_{\bm{x_0}}$-subspace 
    to $\mathcal{O}(d^{q+1})$.
    (c), (d) Energy-density errors $\varepsilon_a\equiv[\widetilde E_a(d)-E_a(\bm{x}_0)]/n$
    versus the sampling distance $d$ at order $q=1$, for the ground state ($a=0$) and the first excited singlet ($a=1$) of the dimerized Heisenberg chain, with the construction applied in the two-dimensional low-energy subspace ($M=2$)
    at several system sizes $n$. Dashed lines are power-law fits, with exponents close to the
    predicted value $2(q+1)=4$.
  }
  \label{fig:d-scaling}
\end{figure}

The central result is the distance scaling law. For given $\bm{x}$, $M$ and $q$ response order, there exists $K$ such that $\mathcal{B}_{M,K}(d)$ provides the following properties. Project $H(x_0)$ in $\mathcal{B}_{M,K}(d)$ and diagonalize it to obtain the first $M$ eigenvalues $\widetilde E_a(d)$.
%Let ${\widetilde E_a(d)}$ be the first $M$ Ritz values of $H(\bm{x}_0)$ in $\mathcal{B}_{M,K}(d)$.
If $H(\bm{x})$ is smooth near $\bm{x}_0$, the external gap in Eq.~\eqref{eq:external-gap} remains open, and the sampling pattern cancels the first $q$ response orders of the $\mathcal{H}_{\bm{x_0}}$-subspace, then for sufficiently small $d$,
\begin{equation}
  0\le
  \frac{\widetilde E_a(d)-E_a(\bm{x}_0)}{n}
  \le
  C_{M,q}^{(n)} d^{2(q+1)},
  \quad
  a=0, \ldots,M-1,
  \label{eq:low-energy-bound}
\end{equation}
where $C_{M,q}^{(n)}$ is independent of $d$.
For a unique gapped ground state in a uniformly gapped local phase, locality further constrains the finite-size prefactor:
\begin{equation}
  0\le
  \frac{\widetilde E_0(d)-E_0(\bm{x}_0)}{n}
  \le
  A_q n^q d^{2(q+1)}
  =
  A_q d^2 (nd^2)^q,
  \label{eq:ground-state-finite-size-bound}
\end{equation}
with $A_q$ independent of $n$ and $d$.

The proof is given in the SM~\cite{SM}. The exponent comes from a simple mechanism.
After cancellation, the approximate low-energy states contain only an $\mathcal{O}(d^{q+1})$ admixture from high-energy states.
This admixture changes the variational energy only at second order, leading to an energy-density error $\mathcal{O}(d^{2(q+1)})$.
The exponent is fixed by the cancellation order alone, while all dependence on system size enters through the prefactor $C_{M,q}^{(n)}$, whose scaling is analyzed next.

Unlike the exponent, the prefactor depends on the system, and how fast it grows decides whether a fixed sampled basis stays accurate at large $n$.
The gap alone does not control this.
For a unique gapped ground state, however, locality constrains it.
%Along a uniformly gapped path, parameter variations are %implemented by a
%quasi-local spectral flow with rapidly decaying spatial tails
%\cite{Hastings2006,bachmann2012automorphic}.
%Combining this quasi-local structure with exponential %clustering~\cite{hastings2005quasiadiabatic}, we show in the %SM~\cite{SM} that each connected
%block of local responses contributes at most one overall volume %factor.
Along a path that remains uniformly gapped, changes in the ground state are generated by a quasi-local spectral flow with rapidly decaying spatial tails~\cite{hastings2005quasiadiabatic,bachmann2012automorphic}.
Together with exponential clustering~\cite{Hastings2006}, this quasi-locality ensures that each connected cluster of response operators contributes at most one factor of the system volume, as shown in the SM~\cite{SM}. 
For the leading residual identified above, the
energy-error expansion contains at most $q+1$ independent connected
clusters. The total-energy error therefore grows at most as
$n^{q+1}d^{2(q+1)}$, and division by the system size gives the $n^q$
prefactor in Eq.~\eqref{eq:ground-state-finite-size-bound}.

According to Eq.~\eqref{eq:ground-state-finite-size-bound}, at $q=0$ a single reference state already gives an error density
$\varepsilon_0\lesssim d^2$, independent of system size.
Each additional order gains a factor $d^2$ in the distance convergence
at the cost of a factor $n$ in the prefactor, so raising $q$
trades a faster approach in $d$ against a stronger growth in $n$.
Higher-order cancellation therefore pays off only when $nd^2\lesssim 1$,
and $nd^2$ sets the crossover scale of the local parameter response.
This combination has a physical interpretation.
It is consistent with the result that the fidelity susceptibility is extensive~\cite{you2007fidelity,campos2007quantum} and the infidelity between two ground states separated by a parameter distance $d$ is of order $nd^2$ for a gapped local ground state.
Thus, nearby points in parameter space need not correspond to nearby many-body states, and $nd^2$ measures their effective separation in Hilbert space.

The bound controls the  $\mathcal{H}_{\bm{x}}$-subspace as a whole.
All retained levels, including the ground state and low-lying excitations,
are bounded with the same exponent $d^{2(q+1)}$,
even when individual eigenvectors are ill-conditioned by near-degeneracies or internal level crossings. This is quite different from the eigenstate continuation, which attempts to 
track individual eigenstate as parameter vary and therefor 
requires a meaningful state-by-state correspondence.
This yields a principled criterion for choosing $M$.
$\mathcal{B}_{M,K}(d)$ should include every level that becomes nearly degenerate
or crosses within the parameter region of interest.

We test these distance scaling laws on a spin-$1/2$ Heisenberg chain of length $n$
with open boundaries and period-four couplings,
\begin{equation}
  H=\sum_{i=1}^{n-1}\mathcal J_i\,\mathbf S_i\cdot\mathbf S_{i+1},
  \qquad \mathcal J_{4r+b}=J_b.
\end{equation}

We set $J_1=J_3=1$ and vary $\bm x=(J_2,J_4)$ around $\bm x_0=(0.5,0.4)$
in the gapped dimerized phase, using reference states from SU(2)-symmetric DMRG~\cite{Schollwck2005,McCulloch2002}.
We apply the construction within the $S=0$ sector and retain the two
lowest singlets as the $\mathcal{H}_{\bm{x}}$-subspace ($M=2$), with the external
gap set by the third singlet level.
The detailed numerical procedure is given in the SM~\cite{SM}.
As shown in Fig.~\ref{fig:d-scaling}(c, d), the energy-density errors
$\varepsilon_a\equiv[\widetilde E_a(d)-E_a(\bm{x}_0)]/n$ for both the ground state and the first excited singlet follow the predicted
$d^{2(q+1)}$ scaling, with fitted exponents close to $4$ for $q=1$ at all system sizes.
The ground state and the excitation obey the same exponent, confirming that the cancellation acts on the $\mathcal{H}_{\bm{x}}$-subspace as a whole rather than on a single eigenvector.

\begin{figure}[t]
  \centering
  \includegraphics[width=\columnwidth]{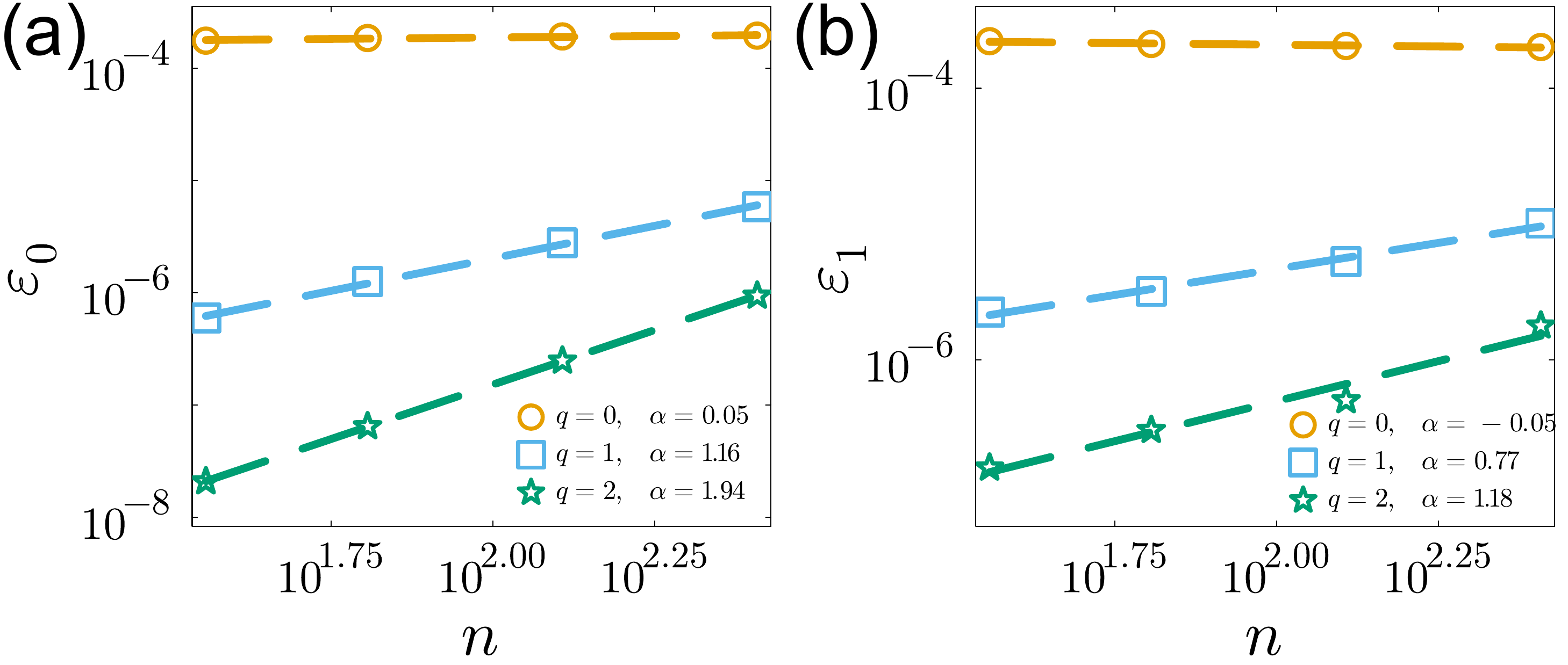}
  \caption{
    Finite-size scaling of the energy density errors.
    (a) Ground state error from $\mathcal{B}_{1,K}(d)$, testing the
    predicted $n^q$ prefactor.
    (b) The first excited singlet error from $\mathcal{B}_{2,K}(d)$.
    Dashed lines are fits to $\varepsilon_a\propto n^\alpha$.
  }
  \label{fig:nq-scaling}
\end{figure}

As presented in Fig.~\ref{fig:nq-scaling}, the ground-state error follows the predicted $n^q$ scaling, with fitted exponents $0.05, 1.16, 1.94$
for $q=0,1,2$. In contrast, the first excited singlet error grows far more slowly, reaching only $1.18$ at $q=2$.
We attribute the contrast to the bulk character of the ground-state response, which accumulates over independent regions, whereas the response to a finite-energy excitation can remain concentrated around the dressed excitation.
%A finite-energy excitation need not accumulate independent contributions
%throughout the bulk. Its response may remain concentrated around the %dressed
%excitation rather than being distributed over the entire system.}
Its finite-size growth is then suppressed, while the distance exponent $d^{2(q+1)}$ stays fixed by the cancellation order.

When the rescaled ground state error $\varepsilon_0/d^2$ is plotted against $nd^2$, data
from different sizes and sampling distances collapse onto one curve,
following $\varepsilon_0/d^2 \sim (nd^2)^q$ until the predicted crossover near $nd^2\sim \mathcal{O}(1)$,
where the growth slows as shown in Fig.~\ref{fig:collapse}(a).
Beyond this point the sampled ground states are no longer globally close to the
target state and can not control the error.
The energy-density error nevertheless saturates, since $\mathcal{B}_{M,K}(d)$ contains the sampled states themselves and any one of them already gives an error density of order $d^2$ at any system size.

\begin{figure}[t]
  \centering
  \includegraphics[width=\columnwidth]{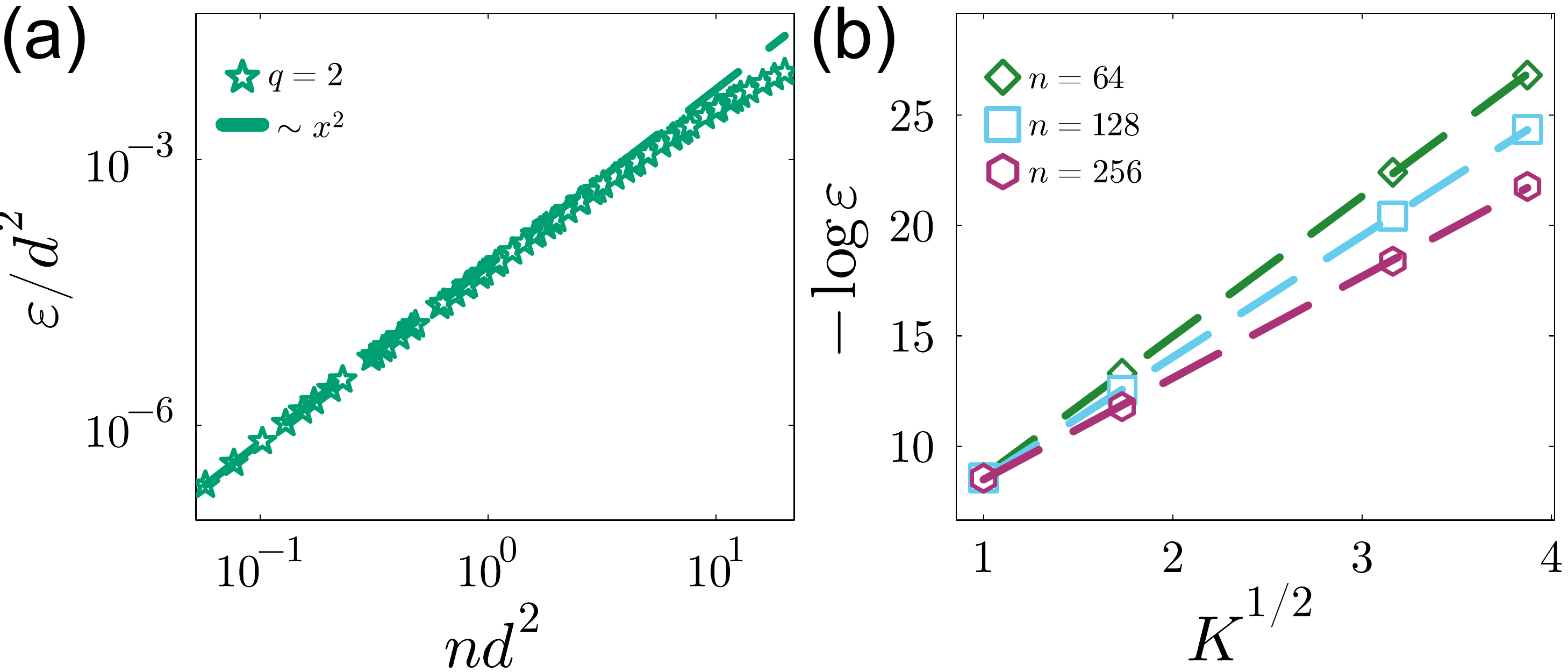}
  \caption{
    Single variable data collapse and sample complexity.
    (a) Rescaled ground-state error $\varepsilon_0/d^2$ versus $nd^2$ for $q=2$,
    combining different $n$ and $d$; the dashed line is the power law $(nd^2)^q$.
    (b) $-\log\varepsilon_0$ versus $\sqrt{K}$ for $n=64,128,256$ at fixed $d$.
  }
  \label{fig:collapse}
\end{figure}

\emph{Sample complexity.---}
Raising the response order improves the distance scaling, but each order costs more sampled parameters.
How much accuracy a given number of samples buys therefore depends on how the coefficient $A_q$ in Eq.~\eqref{eq:ground-state-finite-size-bound} behaves as $q$ grows. We note that increasing $q$ leads to the increasing $K$.
Locality has already accounted for the explicit growth with system size, so what remains in $A_q$ carries no volume factor.
If the response stays analytic in a fixed window around the target, each additional order costs only a local constant.
Together with a stable sequence of sampling patterns, formalized in the SM~\cite{SM}, this gives
\begin{equation}
  \varepsilon_0
  \lesssim
  d^2
  \exp\left[
    -c_D
    |\log(\widetilde C n d^2)|
    K^{1/D}
    \right],
  \label{eq:stretched-exponential-K}
\end{equation}
whenever $\widetilde Cnd^2<1$.
Here $\widetilde C$ is independent of $n$, $d$, and $q$, while $c_D$ depends only on the parameter dimension $D$ and the sampling geometry.
Equivalently, at fixed $D$, $n$, and $d$, reaching accuracy $\varepsilon$ requires
\begin{equation}
  K=\mathcal{O}\!\left([\log(1/\varepsilon)]^D\right)
\end{equation}
sampled parameters.
The exponent $1/D$ in Eq.~\eqref{eq:stretched-exponential-K} arises because cancelling one more response order must be arranged in all $D$ parameter directions at once.
When $\widetilde Cnd^2$ is small, increasing $K$ efficiently raises the response order.
As $\widetilde Cnd^2$ approaches one, the rate in Eq.~\eqref{eq:stretched-exponential-K} vanishes and additional samples give diminishing returns.
Beyond this point the achievable accuracy is limited by the finite overlap between the sampled states and the target state rather than by the response order $q$.
The response order $q$ is not tied to a particular sampling protocol.
More generally, different local sampling geometries can realize the same response order, while their conditioning determines the numerical stability of the construction.
The SM~\cite{SM} compares random sampled parameters with the structured patterns used here.

As a direct test, we vary the number of sampled parameter points.
In our two-dimensional parameter space $(J_2,J_4)$, the number of sampled points scales as $K\sim q^2$, so Eq.~\eqref{eq:stretched-exponential-K} predicts $-\log\varepsilon_0\propto \sqrt K$.
The data in Fig.~\ref{fig:collapse}(b) are consistent with this stretched-exponential improvement.
The curves are approximately linear in $\sqrt{K}$ over the accessible range, with slopes that decrease as the system grows.

\emph{Critical crossover of wavefunction subspace scaling.---} Finally, we examine the construction as the target parameter approaches a continuous transition, where the gapped scaling is expected to degrade with increasing system size.
We characterize this crossover through the convergence exponent.
The bound above relies on an external gap between the retained subspace and the rest of the spectrum.
At a quantum critical point this gap closes in the thermodynamic limit, so no uniform gapped bound can hold across the critical region.
For any finite system, however, the $\mathcal{H}_{\bm{x}}$-subspace remains smooth, and the projected Hamiltonian on $\mathcal{B}_{M,K}$ reflects the finite-size critical response.

To probe the scaling near the critical point, we take the open transverse-field
Ising chain, $H=-\sum_{i=1}^{n-1}\sigma_i^z\sigma_{i+1}^z-g\sum_{i=1}^{n}\sigma_i^x$,
whose transition lies at $g_c=1$ in the thermodynamic limit~\cite{pfeuty1970one}. We retain the two lowest states,
the nearly degenerate Ising doublet in the ferromagnetic phase, as
the low-energy sector, with the protecting gap set by the third level.
We then extract the effective convergence exponent $\beta$ at fixed order $q=1$
as $g$ is tuned through $g_c$. The numerical protocol and fitting procedure are detailed in the
SM~\cite{SM}.

Away from the transition, the effective exponent $\beta$ is close to
the gapped prediction $\beta=2(q+1)=4$. As $g$ approaches $g_c$, however, $\beta$ decreases and develops a pronounced
dip shown in Fig.~\ref{fig:critical}, reflecting a crossover out of the gapped scaling
regime. Near the critical point $g_c$, the dip deepens and sharpens as the system grows, which appears related to
finite-size critical width $|g-g_c|\sim 1/n$.
The convergence exponent thus serves as a finite-size diagnostic of the
critical response of the $\mathcal{H}_{\bm{x}}$-subspace.

\begin{figure}[t]
  \centering
  \includegraphics[width=0.8\columnwidth]{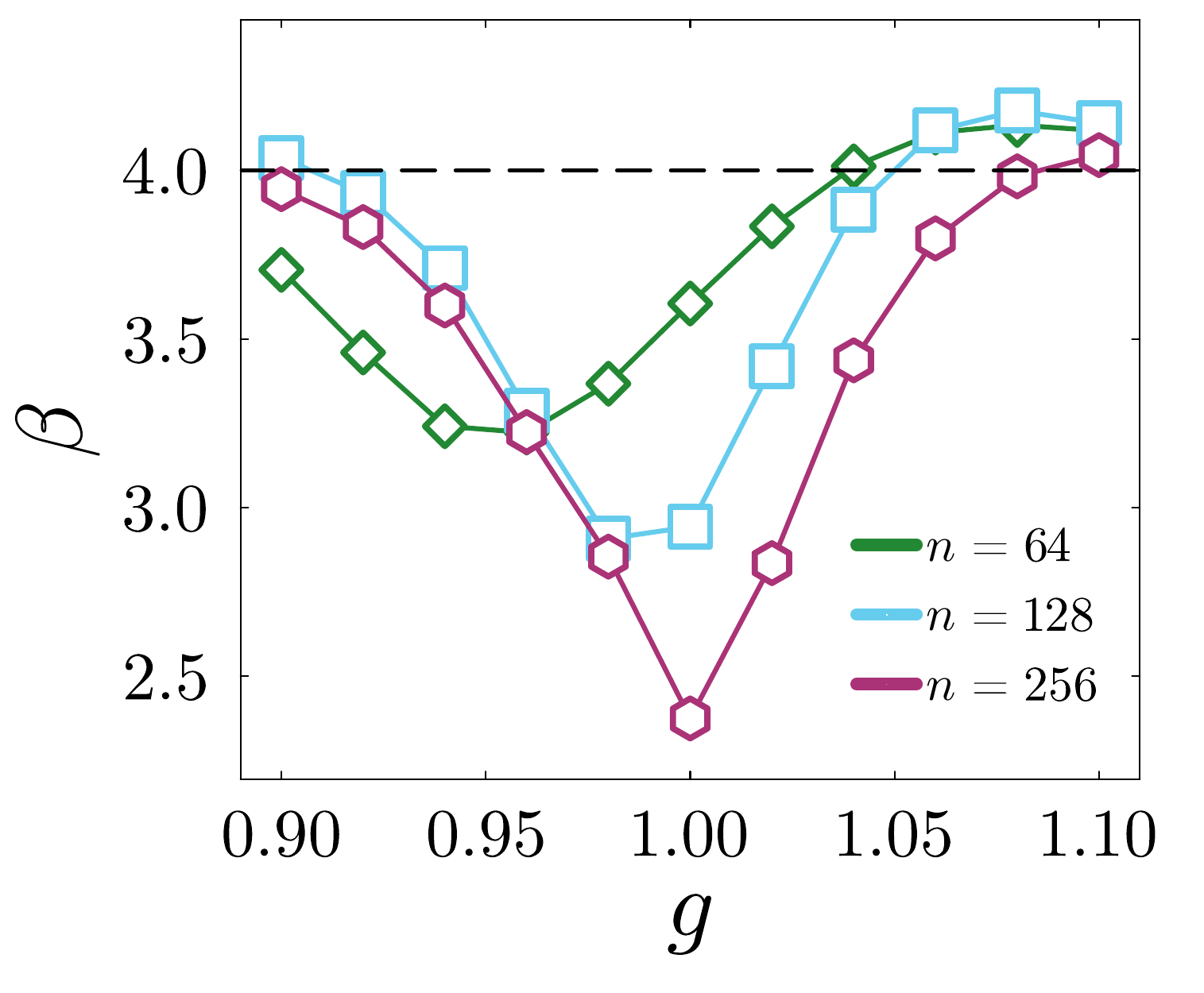}
  \caption{
    Critical crossover of the effective distance scaling exponent.
    The fitted exponent $\beta$ in $\varepsilon\sim d^{\beta}$ is shown as a function
    of the transverse field $g$ for the open transverse field Ising chain.
    The dashed line marks the gapped prediction $\beta=2(q+1)=4$ for $q=1$.
    As the system size increases, $\beta$ develops a deeper and sharper dip near
    $g_c=1$, signaling the crossover out of the gapped scaling regime.
  }
  \label{fig:critical}
\end{figure}

\emph{Conclusions and outlook.---}
We have established a scaling theory for wavefunction transfer based on a low-energy quantum subspace.
The construction guarantees an energy-density error of order $d^{2(q+1)}$, fixed by the cancellation order alone.
For a local gapped ground state, locality controls the finite-size prefactor and identifies $nd^2$ as the natural scaling variable.
When $nd^2\lesssim1$ the ground state is efficiently learnable, with a sampled cost polylogarithmic in the target accuracy in the local analytic regime.
Together, these results turn learning low-energy subspace from a practical heuristic into a quantitatively controlled framework.

Since the theory is formulated at the level of the low-energy subspace, its guarantees remain valid where single-state continuation becomes fragile, including near degeneracies, level crossings, and changes in ground-state character. Near criticality, the convergence exponent departs from its gapped value in a manner governed by the critical response of the low-energy subspace, providing a finite-size probe of the transition. More broadly, the framework is solver-independent and can be incorporated with diverse many-body solvers, including exact diagonalization, tensor networks, neural quantum states, and quantum simulators. It also opens the door to adaptive sampling strategies that place new reference points according to local spectral or wavefunction information, further reducing the cost of exploring parameter space. These extensions establish low-energy subspace learning as a powerful and broadly applicable tool for efficiently resolving quantum states and spectra across families of many-body Hamiltonians.

\emph{Acknowledgments.---}
This work was supported in part by the Beijing Major Science and Technology Project under Contract No. Z251100008125036.
\nocite{Ampelogiannis_arxiv.2405.09388,Hastings2006}
\bibliography{ref}

% ==============================================================================
% =========================== SUPPLEMENTARY MATERIAL ===========================
% ==============================================================================

\clearpage
\onecolumngrid 
\begin{center}
    \textbf{\large Supplemental Material for ``Scaling Theory for Learning Low-Energy Quantum
Subspaces''}
\end{center}

\section{Proof of Theorem~$1$}
\label{sec:proof-main-theorem}

We prove the two parts of Thm.~1 separately.
The first part is a finite-volume variational statement for an isolated low-energy spectral sector.
The second part gives the proof of ground-state finite-size prefactor.

\subsection{Low-energy variational bound}
\label{sec:proof-low-energy}

\begin{assumption}
  \label{ass:sm-low-energy}
  Let $H(\bm{x})$ be a smooth finite-volume Hamiltonian family with system size $n$. Let $D$ be the number of real control parameters,
let $\bm{x}=(x_1,\ldots,x_D)\in\mathbb R^D$ denote the
parameter vector, and fix a target point $\bm{x}_0$.  Let
  \begin{equation*}
    E_0(\bm{x})\le E_1(\bm{x})\le\cdots
  \end{equation*}
  be its eigenvalues. We assume that the first $M$ levels form an isolated
  spectral in a neighborhood of $\bm{x}_0$. Namely,
  \begin{equation*}
    E_M(\bm{x})-E_{M-1}(\bm{x})
    \ge
    \Delta_{\mathrm{high}}>0 .
  \end{equation*}
    Internal degeneracies and crossings inside the first $M$ levels are allowed.
The $M$ retained eigenstates span the low-energy subspace
\begin{equation*}
  \mathcal H(\bm{x})
  =
  \operatorname{span}
  \bigl\{|\psi_a(\bm{x})\rangle\bigr\}_{a=0}^{M-1}.
\end{equation*}
The orthogonal projector onto $\mathcal H(\bm{x})$ is
\begin{equation*}
  P(\bm{x})
  =
  \sum_{a=0}^{M-1}
  |\psi_a(\bm{x})\rangle\langle\psi_a(\bm{x})|.
\end{equation*}
This projector has rank $M$ and is independent of the choice
of orthonormal eigenbasis within the retained state.
The spectral separation ensures that $P(\bm{x})$ depends
smoothly on $\bm{x}$~\cite{kato1966perturbation},
even if the individual eigenvectors cannot be chosen smoothly
across internal level crossings.

Let $K\ge1$ be the number of sampling parameters, and let
$\{\bm{s}_j\}_{j=1}^{K}\subset\mathbb R^D$ be fixed
dimensionless sampling directions.
For a sampling scale $d>0$, define
\begin{equation*}
  \bm{x}_j=\bm{x}_0+d\bm{s}_j,
  \qquad j=1,\ldots,K,
\end{equation*}
with $d$ sufficiently small that all $\bm{x}_j$.
The sampled variational space is
\begin{equation*}
  \cB_{M,K}(d)
  =
  \oplus_{j=1}^K \mathcal H(\bm{x}_j)
  =
  \operatorname{span}
  \bigl\{
    |\psi_a(\bm{x}_j)\rangle
    : a=0,\ldots,M-1,\ j=1,\ldots,K
  \bigr\}.
\end{equation*}

Define a nonnegative integer $q$ as the response order, which can also be used in Taylor expansion as a cancellation order.
Assume that there exist real coefficients $c_1,\ldots,c_K$,
independent of $d$, satisfying
\begin{equation}
  \sum_{j=1}^{K}c_j\bm{s}_j^{\bm{\alpha}}
  =
  \delta_{\bm{\alpha},\bm{0}},
  \qquad |\bm{\alpha}|\le q.
  \label{eq:sm-low-moment}
\end{equation}
Here $\bm{\alpha}=(\alpha_1,\ldots,\alpha_D)$ is a multi-index
of nonnegative integers,
\begin{equation*}
  |\bm{\alpha}|=\sum_{\mu=1}^{D}\alpha_\mu,
  \qquad
  \bm{s}_j^{\bm{\alpha}}
  =\prod_{\mu=1}^{D}(s_{j,\mu})^{\alpha_\mu},
\end{equation*}
where $s_{j,\mu}$ is the $\mu$th component of $\bm{s}_j$.
The symbol $\delta_{\bm{\alpha},\bm{0}}$ equals one when
all components of $\bm{\alpha}$ vanish and zero otherwise.
We will show that Eq.\eqref{eq:sm-low-moment} provides the moment conditions that reproduce the constant term
and cancel all Taylor terms of total degree $1,\ldots,q$.
\end{assumption}
Let
\begin{equation*}
  H_0=H(\bm{x}_0),\qquad
  P_0=P(\bm{x}_0),\qquad
  Q_0=I-P_0,\qquad
  \cL_0=\mathcal H_{\bm{x}_0},
\end{equation*}
where $I$ is the identity operator on the full Hilbert space.
Thus $Q_0$ is the orthogonal projector onto the orthogonal
complement of the target low-energy subspace $\cL_0$.

Assumption~\ref{ass:sm-low-energy} ensures that $P(\bm{x})$
is a smooth rank-$M$ spectral projector near $\bm{x}_0$.
After shrinking the neighborhood if necessary, one can choose
a smooth family of unitary operators $U(\bm{x})$ on the full
Hilbert space such that
\begin{equation*}
  U(\bm{x}_0)=I,
  \qquad
  P(\bm{x})=U(\bm{x})P_0U(\bm{x})^\dagger,
\end{equation*}
where the dagger denotes the Hermitian adjoint.
The restriction of $U(\bm{x})$ to $\cL_0$ therefore gives
an isometric identification
\begin{equation*}
  U(\bm{x})|_{\cL_0}:
  \cL_0\longrightarrow\mathcal H_{\bm{x}}.
\end{equation*}
This identification is a choice of frame for the low-energy
subspaces; its basis vectors need not be individual
eigenvectors of $H(\bm{x})$.
The sampled variational space is independent of this choice.

Using the sampling points $\bm{x}_j=\bm{x}_0+d\bm{s}_j$
and the coefficients $c_j$ from
Eq.~\eqref{eq:sm-low-moment}, define
\begin{equation*}
  T_d=\sum_{j=1}^{K}c_jU(\bm{x}_j)P_0.
\end{equation*}
For every $u\in\cL_0$, each vector $U(\bm{x}_j)u$
belongs to $\mathcal H(\bm{x}_j)$, and hence
\begin{equation*}
  T_d\cL_0\subset\cB_{M,K}(d).
\end{equation*}

Taylor expansion of $U(\bm{x})u$ gives

\begin{equation*}
  U(\bm{x}_0+d\bm{s}_j)u
  =
  \sum_{|\bm{\alpha}|\le q}
  \frac{d^{|\bm{\alpha}|}\bm{s}_j^{\bm{\alpha}}}{\bm{\alpha}!}
  \partial^{\bm{\alpha}}\!\left[U(\bm{x})u\right]_{\bm{x}=\bm{x}_0}
  +
  R_j(u,d),
\end{equation*}
with $\norm{R_j(u,d)}\le C_jd^{q+1}\norm{u}$ for sufficiently small $d$. 
Here $R_j(u,d)$ is the Taylor remainder, and $C_j$ is a
finite constant independent of $u$ and $d$ at fixed volume.

Multiplying by $c_j$ and summing over $j$,
Eq.~\eqref{eq:sm-low-moment} cancels all terms with $1\le |\bm{\alpha}|\le q$ and leaves the zeroth-order term. Thus

\begin{equation}
  T_du=u+r_d(u),
  \qquad
  \norm{r_d(u)}\le C_T^{(n)}d^{q+1}\norm{u}.
  \label{eq:sm-low-subspace-error}
\end{equation}
where 
\begin{equation*}
  r_d(u)=\sum_{j=1}^{K}c_jR_j(u,d),
  \qquad
  C_T^{(n)}=\sum_{j=1}^{K}|c_j|C_j.
\end{equation*}
The constant $C_T^{(n)}$ is independent of $d$ at fixed finite volume,
but may depend on $n$. Therefore $T_d$ is injective on $\cL_0$ for small enough $d$. The subspace

\begin{equation*}
  \cL_d=T_d\cL_0
\end{equation*}
has dimension $M$ and is contained in $\cB_{M,K}(d)$.

We write $\cL_d$ as a graph over $\cL_0$. Let
\begin{equation*}
  A_d=P_0T_dP_0:\cL_0\to\cL_0 .
\end{equation*}
Eq.~\eqref{eq:sm-low-subspace-error} gives $A_d=P_0+O(d^{q+1})$, so $A_d$ is invertible for small $d$. Define

\begin{equation*}
  X_d=Q_0T_dP_0A_d^{-1}:\cL_0\to Q_0\cH .
\end{equation*}
Then
\begin{equation}
  \cL_d=\{u+X_du:u\in\cL_0\},
  \qquad
  \norm{X_d}\le C_X^{(n)}d^{q+1}.
  \label{eq:sm-low-graph}
\end{equation}

Let
\begin{equation*}
  W_d=(1+X_d)(1+X_d^\dagger X_d)^{-1/2}
\end{equation*}
be the canonical isometry from $\cL_0$ to $\cL_d$. The Hamiltonian restricted to $\cL_d$ and pulled back to $\cL_0$ is
\begin{equation*}
  H_{\mathrm{eff}}(d)=W_d^\dagger H_0W_d .
\end{equation*}
Since $P_0$ is a spectral projector of $H_0$, one has $P_0H_0Q_0=Q_0H_0P_0=0$. Consequently
\begin{equation}
  H_{\mathrm{eff}}(d)
  =
  (1+X_d^\dagger X_d)^{-1/2}
  \left[
    P_0H_0P_0
    +
    X_d^\dagger Q_0H_0Q_0X_d
    \right]
  (1+X_d^\dagger X_d)^{-1/2}.
  \label{eq:sm-low-Heff-graph}
\end{equation}
The correction is quadratic in $X_d$. This is the origin of the exponent $2(q+1)$ in the energy bound.

At fixed finite volume, define
\begin{equation*}
  C_{M,q}^{(n)}
  =
  \sup_{0<d\le d_0}
  d^{-2(q+1)}
  \frac{1}{n}
  \left\|
  H_{\mathrm{eff}}(d)-P_0H_0P_0
  \right\|_{\cL_0\to\cL_0}.
\end{equation*}
For sufficiently small $d_0$, this quantity is finite and independent of $d$. Indeed, Eq.~\eqref{eq:sm-low-graph}
and Eq.~\eqref{eq:sm-low-Heff-graph} give
\begin{equation*}
  \left\|
  H_{\mathrm{eff}}(d)-P_0H_0P_0
  \right\|_{\cL_0\to\cL_0}
  =
  O(d^{2(q+1)})
\end{equation*}
at fixed finite volume. The constant may depend on $n$.

Let $\varepsilon_0(d)\le\cdots\le\varepsilon_{M-1}(d)$ be the eigenvalues of
$H_{\mathrm{eff}}(d)$.
The eigenvalues of $P_0H_0P_0$ on $\cL_0$ are $E_0(\bm{x}_0),\ldots,E_{M-1}(\bm{x}_0)$.
Weyl's inequality gives
\begin{equation*}
  \frac{
    |\varepsilon_a(d)-E_a(\bm{x}_0)|
  }{n}
  \le
  C_{M,q}^{(n)}d^{2(q+1)},
  \qquad
  a=0,\ldots,M-1 .
\end{equation*}
Let $\widetilde E_a(d)$ be the Ritz values obtained by diagonalizing $H_0$ in the full sampled space $\cB_{M,K}(d)$.
Since $\cL_d\subset\cB_{M,K}(d)$, the min--max principle gives
\begin{equation*}
  \widetilde E_a(d)\le \varepsilon_a(d),
  \qquad
  a=0,\ldots,M-1 .
\end{equation*}
The same principle applied to the full Hamiltonian gives
\begin{equation*}
  E_a(\bm{x}_0)\le \widetilde E_a(d),
  \qquad
  a=0,\ldots,M-1 .
\end{equation*}
Combining the last three estimates yields
\begin{equation*}
  0\le
  \frac{\widetilde E_a(d)-E_a(\bm{x}_0)}{n}
  \le
  C_{M,q}^{(n)}d^{2(q+1)},
  \qquad
  a=0,\ldots,M-1 .
\end{equation*}
This proves the low-energy part of Thm.~1.

\subsection{Ground-state prefactor}
\label{sec:proof-ground-prefactor}
We prove the finite-size prefactor estimate in the second part of Thm.~1 from
a volume-uniform gap, locality, and finite-order smoothness.
Throughout this subsection, $\Lambda$ is a finite lattice, $n=|\Lambda|$, and $m=q+1$.

\begin{assumption}
  \label{ass:sm-ground-local-family}
 let
  \begin{equation*}
    H_\Lambda(\bm{x})=\sum_{Z\subset\Lambda}h_Z(\bm{x})
  \end{equation*}
where $\Lambda$ is a finite lattice and $h_Z(\bm{x})$
is an interaction term acting on the sites in $Z$. For every $\bm{x}\in V$, let
$|\psi_{0,\Lambda}(\bm{x})\rangle$ be the unique
normalized ground state of $H_\Lambda(\bm{x})$,
with energy $E_{0,\Lambda}(\bm{x})$.
We assume
\begin{equation*}
  E_{1,\Lambda}(\bm{x})-E_{0,\Lambda}(\bm{x})
  \ge \Delta>0,
\end{equation*}
where $E_{1,\Lambda}(\bm{x})$ is the lowest excited-state
energy and $\Delta$ is independent of $\Lambda$ and
$\bm{x}\in V$.

The interaction has continuous parameter derivatives
through order $m+1$. We assume that, for every fixed
$p\ge0$ and $|\bm{\alpha}|\le m+1$,
\begin{equation*}
  \sup_{\Lambda,\,\bm{x}\in U,\,i\in\Lambda}
  \sum_{\substack{Z\subseteq\Lambda\\Z\ni i}}
  \left(1+|Z|+\operatorname{diam}(Z)\right)^p
  \bigl\|\partial^{\bm{\alpha}}h_Z(\bm{x})\bigr\|
  \le C_{\bm{\alpha},p}<\infty.
\end{equation*}
Here $i$ labels a lattice site, $|Z|$ is the number of
sites in $Z$, $\operatorname{diam}(Z)$ is their maximum
distance in lattice units, and $\|\cdot\|$ denotes
the operator norm.
The constants $C_{\bm{\alpha},p}$ are independent of
$\Lambda$ and $\bm{x}\in U$.
This condition expresses that the interaction and its
parameter derivatives remain sufficiently short-ranged,
uniformly in the system size.
It includes finite-range and exponentially decaying
interactions when the same uniform locality bounds
hold for their parameter derivatives.
\end{assumption}

For finite-range or exponentially decaying interactions,
the uniform spectral gap implies exponential clustering:
connected correlations between local observables decay
exponentially with their spatial separation, with bounds
uniform in $\Lambda$ and $\bm{x}\in U$~\cite{Hastings2006}.

Together with locality and smooth parameter dependence,
the gap also allows the change of the ground-state projector
to be generated by a quasi-local Hermitian operator
$K_{\bm v,\Lambda}(\bm{x})$
~\cite{bachmann2012automorphic}:
\begin{equation}
  \partial_{\bm v}P_\Lambda(\bm{x})
  =
  -i[K_{\bm v,\Lambda}(\bm{x}),P_\Lambda(\bm{x})],
  \qquad
  P_\Lambda(\bm{x})
  =
  |\psi_{0,\Lambda}(\bm{x})\rangle
  \langle\psi_{0,\Lambda}(\bm{x})|.
  \label{eq:sm-ground-spectral-flow}
\end{equation}
Here $\bm v=(v_1,\ldots,v_D)$ is a fixed direction in
parameter space,
\begin{equation*}
  \partial_{\bm v}
  =
  \sum_{\mu=1}^{D}v_\mu\partial_{x_\mu},
\end{equation*}
The operator $K_{\bm v,\Lambda}$ generates transport
in parameter space, rather than physical time evolution.

For each nonnegative derivative order $a$ needed below,
write
\begin{equation*}
  K^{(a)}_{\bm v,\Lambda}(\bm{x})
  =
  \partial_{\bm v}^{\,a}K_{\bm v,\Lambda}(\bm{x}),
\end{equation*}
with $a=0$ denoting the generator itself.
Under the assumed locality and smoothness conditions,
these operators admit local decompositions
\begin{equation*}
  K^{(a)}_{\bm v,\Lambda}(\bm{x})
  =
  \sum_{y\in\Lambda}\sum_{\ell=0}^{\infty}
  k^{(a)}_{\bm v;y,\ell}(\bm{x}),
  \qquad
  \operatorname{supp}k^{(a)}_{\bm v;y,\ell}
  \subseteq B_\ell(y).
\end{equation*}
Here $y$ labels a lattice site, and $\ell$ is a nonnegative
integer radius in lattice units.
The notation $\operatorname{supp}k^{(a)}_{\bm v;y,\ell}$
denotes the set of sites on which
$k^{(a)}_{\bm v;y,\ell}$ acts nontrivially, and
\begin{equation*}
  B_\ell(y)
  =
  \{z\in\Lambda:\operatorname{dist}(y,z)\le\ell\}
\end{equation*}
is the lattice ball centered at $y$, where
$\operatorname{dist}(y,z)$ is the lattice distance
between sites $y$ and $z$.
The dependence of the local terms on $\Lambda$ is
suppressed in the notation.

The quasi-locality of the spectral-flow generator and
its finite-order parameter derivatives follows from
the locality estimates of
Ref.~\cite{bachmann2012automorphic}, for any given finite $p>0$
\begin{equation}
  \sup_{\Lambda,\,\bm{x}\in U,\,y\in\Lambda}
  \sum_{\ell=0}^{\infty}
  (1+\ell)^p
  \bigl\|k^{(a)}_{\bm v;y,\ell}(\bm{x})\bigr\|
  \le C_{a,p,\bm v}<\infty.
  \label{eq:sm-ground-rapid-tail}
\end{equation}
The constants are independent of the system size and
can be chosen uniformly for $\bm v$ in a bounded set.
Physically, this means that the generator and its
required derivatives have rapidly decaying spatial tails.
Only finitely many derivative orders $a$ and weights $p$,
depending on the fixed response order $q$, are used below.

We now introduce the connected correlations needed
to count powers of the system size.
For an observable $A$, denote its ground-state expectation
value by
\begin{equation*}
  \langle A\rangle_{\bm{x}}
  =
  \langle\psi_{0,\Lambda}(\bm{x})|
  A|\psi_{0,\Lambda}(\bm{x})\rangle,
\end{equation*}
with the volume label $\Lambda$ suppressed.

For $r\ge1$ possibly noncommuting observables
$O_1,\ldots,O_r$, define the ordered cumulant
\begin{equation}
  \kappa_{\bm{x}}(O_1,\ldots,O_r)
  =
  \sum_{\pi\in\Pi_r}
  (-1)^{|\pi|-1}(|\pi|-1)!
  \prod_{B\in\pi}
  \left\langle
    \prod_{i\in B}^{\longrightarrow}O_i
  \right\rangle_{\bm{x}}.
  \label{eq:sm-ordered-cumulant}
\end{equation}
Here $\Pi_r$ is the set of all partitions of
$\{1,\ldots,r\}$ into nonempty disjoint blocks.
For a partition $\pi$, $|\pi|$ is its number of blocks,
and $B$ denotes one such block.
The arrow means that operators within each block
retain their original order:
if $B=\{i_1<\cdots<i_s\}$, then
\begin{equation*}
  \prod_{i\in B}^{\longrightarrow}O_i
  =
  O_{i_1}\cdots O_{i_s}.
\end{equation*}
The product over blocks is a product of scalar
expectation values, so its order does not matter.
The cumulant extracts the connected part of a correlation
function by subtracting all contributions that factorize
into lower-order correlations.
For example, the second-order cumulant is
\begin{equation*}
  \kappa_{\bm{x}}(O_1,O_2)
  =
  \langle O_1O_2\rangle_{\bm{x}}
  -
  \langle O_1\rangle_{\bm{x}}
  \langle O_2\rangle_{\bm{x}},
\end{equation*}
which is the usual connected two-point correlation.
At higher orders, the same subtraction removes products
of correlations within smaller groups of operators.
Consequently, if the moments factorize between two
independent groups, any cumulant involving operators
from both groups vanishes.
For noncommuting observables, the order of operators
within each moment is retained as specified above.

This connected structure is useful for volume counting:
when connected correlations are sufficiently short-ranged,
summing over relative positions gives a finite contribution,
while the overall position contributes a single factor
of the system size $n$.

\begin{lemma}
\label{lem:sm-ground-connected-block}
Fix a positive integer $R$, the maximum cumulant order,
and let $1\le r\le R$.
Consider $r$ operators with local decompositions
\begin{equation*}
  A_j
  =
  \sum_{y\in\Lambda}\sum_{\ell=0}^{\infty}b_{j;y,\ell},
  \qquad j=1,\ldots,r,
\end{equation*}
where $b_{j;y,\ell}$ is supported within the ball
$B_\ell(y)$ and satisfies the rapid-tail bound
in Eq.~\eqref{eq:sm-ground-rapid-tail},
with constants uniform in $\Lambda$ and $\bm{x}\in V$.
Then
\begin{equation}
  \left|
    \kappa_{\bm{x}}(A_1,\ldots,A_r)
  \right|
  \le C_R n,
  \label{eq:sm-ground-extensive-cumulant}
\end{equation}
uniformly in $\Lambda$ and $\bm{x}\in U$. Here $\kappa_{\bm{x}}$ is defined in Eq.~\eqref{eq:sm-ordered-cumulant}.
The constant $C_R$ may depend on $R$, the lattice geometry,
the clustering bounds, and the local decay bounds,
but not on $\Lambda$ or $\bm{x}$.

The same conclusion holds when one of the operators
$A_j$ is replaced by the local Hamiltonian
$H_\Lambda(\bm{x}')$ at a fixed parameter point
$\bm{x}'\in U$.
\end{lemma}

\begin{proof}
  We include the summability argument because a two-point estimate alone does
  not immediately give Eq.~\eqref{eq:sm-ground-extensive-cumulant}.
  First take local observables $o_a$ supported on finite sets $X_a$.  Let
  $L(X_1,\ldots,X_r)$ be the length of the longest edge in a minimum spanning tree whose
  vertices are the supports and whose edge weights are support distances.
  We claim, by induction on $r$, that
  \begin{equation}
    \left|\kappa_{\bm{x}}(o_1,\ldots,o_r)\right|
    \le
    C_r
    \left(\sum_{a=1}^r|X_a|\right)^{b_r}
    \prod_{a=1}^r\norm{o_a}
    e^{-\mu_r L(X_1,\ldots,X_r)},
    \label{eq:sm-ground-longest-edge-bound}
  \end{equation}
  for constants $b_r<\infty$ and $\mu_r>0$ that may depend on $r$ but not on
  the volume.  For $r=2$, this is exponential clustering.  Suppose the claim
  holds below order $r$.  Remove a longest edge of a minimum spanning tree,
  producing a cut $  I\cup J=\{1,\ldots,r\}, \quad I \cap J=\varnothing$.  Every support on one side is
  separated from every support on the other by at least the removed edge
  length; otherwise the tree could be reconnected across a shorter edge.
  Since separated local observables commute, the original ordered product can
  be grouped into the two ordered products
  \begin{equation*}
    O_I=\prod_{a\in I}^{\longrightarrow}o_a,
    \qquad
    O_J=\prod_{a\in J}^{\longrightarrow}o_a .
  \end{equation*}

  Define the fluctuation operators
\begin{equation*}
  \delta O_I=O_I-\langle O_I\rangle_{\bm{x}},
  \qquad
  \delta O_J=O_J-\langle O_J\rangle_{\bm{x}},
\end{equation*}
Their connected correlation is
\begin{equation*}
  \langle\delta O_I\,\delta O_J\rangle_{\bm{x}}
  =
  \langle O_I O_J\rangle_{\bm{x}}
  -
  \langle O_I\rangle_{\bm{x}}\langle O_J\rangle_{\bm{x}}.
\end{equation*}
  Two-point clustering applied to $O_I$ and $O_J$
bounds their connected correlation by
  the right-hand side of Eq.~\eqref{eq:sm-ground-longest-edge-bound}, with the
  support-size prefactor replaced by a fixed polynomial in
  $\sum_a|X_a|$.

  Expand this connected correlation using the
moment--cumulant relation.  Partitions whose
  blocks lie entirely in $I$ or entirely in $J$ cancel exactly against
  $\langle O_I\rangle\langle O_J\rangle$.  Every remaining partition has a
  block meeting both sides of the cut.  The one-block partition is the desired
  $r$-point cumulant.  In every other remaining partition, each crossing block
  has order strictly smaller than $r$ and therefore obeys the induction
  hypothesis with an exponential factor in the cut length.  All noncrossing
  blocks are bounded by the corresponding moment bounds.  The number of
  partitions is finite at fixed $r$.  After increasing $b_r$ and decreasing
  $\mu_r$ if necessary, this proves
  Eq.~\eqref{eq:sm-ground-longest-edge-bound}.  The possible degradation of
  $\mu_r$ is harmless because $r\le R$ is fixed.

  We now insert the ball decompositions.  Fixing the center of one ball leaves
  only relative center positions and radii to sum.  At fixed radii, the number
  of configurations with longest tree edge of order $L$ grows only
  polynomially in $L$, because the lattice has polynomial volume growth and
  $r$ is fixed.  The exponential factor in
  Eq.~\eqref{eq:sm-ground-longest-edge-bound} makes this relative-position sum
  finite.  If large balls bridge distant centers, the required support-size
  factors and ball radii are absorbed by choosing the polynomial weight $p$ in
  Eq.~\eqref{eq:sm-ground-rapid-tail} sufficiently large.  Thus all relative
  sums are bounded uniformly, while the first ball center has at most $n$
  choices.  This proves Eq.~\eqref{eq:sm-ground-extensive-cumulant}.  A local
  Hamiltonian insertion has the same weighted decomposition and is treated
  identically.  Scalar shifts of that insertion do not affect cumulants of
  order at least two.
\end{proof}

The preceding proof is a finite-order statement.  It neither requires a
rate uniform in $r$ nor asserts a cumulant bound involving several parameter
points.  General inheritance of clustering by higher connected correlations
is discussed in Ref.~\cite{Ampelogiannis_arxiv.2405.09388}; the explicit
longest-edge induction above supplies the finite-volume summability needed
here.  Below we use Lemma~\ref{lem:sm-ground-connected-block} only through
order $R=2m+1$: at most $2m$ response insertions and one Hamiltonian
insertion occur.

We next fix the ground-state gauge.  Put
\begin{equation}
  \widehat H_{0,\Lambda}
  =H_\Lambda(\bm{x}_0)-E_{0,\Lambda}(\bm{x}_0)\ge0,
  \qquad
  \norm{u}_{\widehat H_0}^2
  =\bra{u}\widehat H_{0,\Lambda}\ket{u}.
  \label{eq:sm-ground-energy-seminorm}
\end{equation}

For a direction $\bm v\in\mathbb R^D$, consider the radial
path $\bm{x}_0+t\bm v$, with $t$ chosen so that the path
remains in $V$.
Let $U_\Lambda(t,\bm v)$ be the unitary transport operator
along this path.
It satisfies
\begin{equation}
  \partial_t U_\Lambda(t,\bm{v})
  =
  -iK_{\bm{v},\Lambda}(\bm{x}_0+t\bm{v})
  U_\Lambda(t,\bm{v}),
  \qquad
  U_\Lambda(0,\bm{v})=\id.
  \label{eq:sm-ground-radial-transport}
\end{equation}
We choose the scalar part of the generator to fix the
ground-state phase by parallel transport along each
radial path.
Under Assumption~\ref{ass:sm-ground-local-family},
the spectral-flow generator depends smoothly on the
path parameter $t$ and direction $\bm v$, with continuous
derivatives through order $m$.
The transported state therefore also varies smoothly:
\begin{equation*}
  |\psi_{0,\Lambda}(\bm{x}_0+t\bm v)\rangle
  =
  U_\Lambda(t,\bm v)
  |\psi_{0,\Lambda}(\bm{x}_0)\rangle.
\end{equation*}
This construction fixes the phase of the ground-state
vector throughout the neighborhood, with continuous
parameter derivatives through order $m$,
including at $\bm{v}=0$, and the transported state satisfies
\begin{equation}
  \left\langle
    \psi_{0,\Lambda}(\bm{x}_0+t\bm{v})
    \,\middle|\,
    \partial_t\psi_{0,\Lambda}(\bm{x}_0+t\bm{v})
  \right\rangle
  =0.
  \label{eq:sm-ground-parallel-gauge}
\end{equation}

In the chosen gauge, the ground-state has
continuous parameter derivatives through order $m$.
Consequently, mixed derivatives commute through order $m$.
For any direction $\bm v\in\mathbb R^D$ and every
integer $1\le r\le m$,
\begin{equation}
  \left.
    \partial_{\bm v}^{\,r}
    |\psi_{0,\Lambda}(\bm{x})\rangle
  \right|_{\bm{x}=\bm{x}_0}
  =
  \sum_{|\bm{\alpha}|=r}
  \frac{r!}{\bm{\alpha}!}\bm v^{\bm{\alpha}}
  \left.
    \partial^{\bm{\alpha}}
    |\psi_{0,\Lambda}(\bm{x})\rangle
  \right|_{\bm{x}=\bm{x}_0}.
  \label{eq:sm-ground-directional-polynomial}
\end{equation}

Before estimating higher-order derivatives of the
ground state, we establish a uniform energy bound
for a single sampled state.
For a fixed bounded direction $\bm v$, set
\begin{equation*}
  f_{\bm v}(t)
  =\bra{\psi_{0,\Lambda}(\bm{x}_0+t\bm v)}
  \widehat H_{0,\Lambda}
  \ket{\psi_{0,\Lambda}(\bm{x}_0+t\bm v)} .
\end{equation*}
Writing $K(t)=K_{\bm v,\Lambda}(\bm{x}_0+t\bm v)$ and differentiating along
the spectral flow gives
\begin{equation}
  f_{\bm v}''(t)
  =i\left\langle[K'(t),\widehat H_{0,\Lambda}]\right\rangle_t
  -\left\langle[K(t),[K(t),\widehat H_{0,\Lambda}]]\right\rangle_t .
  \label{eq:sm-ground-snapshot-second-derivative}
\end{equation}
where $ K'(t)=\frac{dK(t)}{dt}$.
The interaction-norm commutator estimates, together with the Lemma.~ on $K$ and $K'$, imply
\begin{equation*}
  \sup_{\Lambda,\,|t|\le t_0,\,\norm{\bm v}\le S}
  |f_{\bm v}''(t)|\le C_S n
\end{equation*}
Here $S>0$ is the maximum allowed norm of the directions.
We choose $t_0>0$, independently of the volume, so that
$\bm{x}_0+t\bm v\in V$ whenever
$|t|\le t_0$ and $\|\bm v\|\le S$. This is a locality estimate:
the commutators in Eq.~\eqref{eq:sm-ground-snapshot-second-derivative} are
again extensive quasi-local interactions.  Since
$f_{\bm v}(0)=f_{\bm v}'(0)=0$, integral Taylor expansion yields
\begin{equation}
  0\le \frac{f_{\bm v}(t)}{n}\le C_S t^2 .
  \label{eq:sm-ground-single-snapshot}
\end{equation}

\begin{lemma}
\label{lem:sm-ground-jet-bounds}
Fix an integer $1\le r\le m$ and a direction $\bm v$
with $\|\bm v\|\le S$.
Along the ray $\bm{x}=\bm{x}_0+t\bm v$ within $V$,
\begin{equation}
  \|\partial_{\bm v}^{\,r}\psi_{0,\Lambda}(\bm{x})\|^2
  \le C_r n^r.
  \label{eq:sm-ground-jet-norm}
\end{equation}
For a fixed constant $\rho>0$, if additionally
$n\|\bm{x}-\bm{x}_0\|^2\le\rho$, then
\begin{equation}
  \|\partial_{\bm v}^{\,r}\psi_{0,\Lambda}(\bm{x})\|_{\widehat H_0}^2
  \le C_r n^r.
  \label{eq:sm-ground-jet-energy}
\end{equation}
The constants may depend on $r$, $S$, and $\rho$,
but are independent of the volume and uniform
over the indicated parameter points and directions.
\end{lemma}

\begin{proof}
  Fix a base point $\bm{x}$ on the ray and introduce two scalar displacements
  $u,v$ along the same ray:
  \begin{align}
    G_{\bm{x}}(u,v)
     & =\bra{\psi_{0,\Lambda}(\bm{x}+u\bm v)}
    \psi_{0,\Lambda}(\bm{x}+v\bm v)\rangle,
    \label{eq:sm-ground-overlap-kernel}
    \\
    F_{\bm{x}}(u,v)
     & =\bra{\psi_{0,\Lambda}(\bm{x}+u\bm v)}
    \widehat H_{0,\Lambda}
    \ket{\psi_{0,\Lambda}(\bm{x}+v\bm v)}.
    \label{eq:sm-ground-energy-kernel}
  \end{align}

Repeated differentiation of the transport equation
expresses the Taylor coefficients of $G_{\bm{x}}(u,v)$
as finite sums of ordered expectation values involving
the spectral-flow generator and its parameter derivatives.
The coefficients of $F_{\bm{x}}(u,v)$ contain one
additional insertion of $\widehat H_{0,\Lambda}$.
The operator ordering follows from differentiating
the bra and ket in these expressions.

Since
$G_{\bm{x}}(0,0)=1$, near $u=0,\quad v=0$, we can define
\begin{equation*}
  L_{\bm{x}}=\log G_{\bm{x}},
  \qquad
  J_{\bm{x}}=\frac{F_{\bm{x}}}{G_{\bm{x}}}
\end{equation*}
Taking the logarithm removes contributions that factorize
into separate correlation blocks.
Similarly, dividing $F_{\bm{x}}$ by $G_{\bm{x}}$
removes factors disconnected from the insertion of
$\widehat H_{0,\Lambda}$.

The Taylor coefficients of $L_{\bm{x}}$ are finite sums
of connected correlations involving the quasi-local
generator and its parameter derivatives.
The nonconstant coefficients of $J_{\bm{x}}$ have the
same structure, with each connected correlation also
containing an insertion of $\widehat H_{0,\Lambda}$.

Let $[u^av^b]$ denote the coefficient of $u^av^b$
in the Taylor expansion about $u=v=0$, where $a,b$
are nonnegative integers.
Applying Lemma~\ref{lem:sm-ground-connected-block}
at the base point $\bm{x}$ gives
\begin{align*}
  \left|[u^av^b]L_{\bm{x}}\right|
  &\le C_{a,b}n,
  && a+b\ge2, \\
  \left|[u^av^b]J_{\bm{x}}\right|
  &\le C_{a,b}n,
  && a+b\ge1,
\end{align*}
at all fixed response orders required below.
The constants may depend on $a,b$ and the fixed
direction bound, but are independent of the volume
and uniform over the parameter region under consideration.

The constant term
\begin{equation*}
  J_{\bm{x}}(0,0)
  =
  \langle\widehat H_{0,\Lambda}\rangle_{\bm{x}}
\end{equation*}
is estimated separately below.

Eq.~\eqref{eq:sm-ground-parallel-gauge} makes the two
degree-one coefficients of $L_{\bm{x}}$ vanish.
Since $L_{\bm{x}}(0,0)=0$, each factor from $L$
contributes at least two powers of $u,v$ and at most
one power of $n$. Therefore,
\begin{equation*}
  |[u^av^b]G_{\bm{x}}|
  \le
  \begin{cases}
    C_{a,b}n^{(a+b)/2},
      & a+b \text{ even},\\
    C_{a,b}n^{(a+b-1)/2},
      & a+b \text{ odd}.
  \end{cases}
\end{equation*}
Taking $a=b=r$ gives
\begin{equation*}
  \|\partial_{\bm v}^{\,r}\psi_{0,\Lambda}(\bm{x})\|^2
  =
  (r!)^2[u^rv^r]G_{\bm{x}}
  \le C_rn^r,
\end{equation*}
proving Eq.~\eqref{eq:sm-ground-jet-norm}.

  At $\bm{x}_0$, the relation
  $\widehat H_{0,\Lambda}\ket{\psi_{0,\Lambda}(\bm{x}_0)}=0$ makes the
  constant and degree-one coefficients of $J_{\bm{x}_0}$ vanish.  With
  $\delta=\norm{\bm{x}-\bm{x}_0}$, the uniform linked-coefficient bounds,
  Eq.~\eqref{eq:sm-ground-single-snapshot}, and the fundamental theorem of
  calculus along the ray give
  \begin{align*}
    |J_{\bm{x}}(0,0)|    & \le Cn\delta^2,
    \\
    |[u^av^b]J_{\bm{x}}| & \le Cn\delta,
                         &                 & a+b=1,
    \\
    |[u^av^b]J_{\bm{x}}| & \le C_{a,b}n,
                         &                 & a+b\ge2.
  \end{align*}
  In the coefficient $u^rv^r$ of $F=GJ$, the degree-zero part of $J$
  contributes at most
  $Cn\delta^2 n^r\le C\rho n^r$.  A degree-one part of $J$ multiplies a
  $G$ coefficient of total degree $2r-1$, which has at most $r-1$ linked
  factors, and contributes at most $Cn\delta n^{r-1}\le Cn^r$.  Every term
  of $J$ of degree at least two also multiplies at most $r-1$ linked factors
  from $G$ and is $O(n^r)$.  Thus
\begin{equation*}
  \left.
    \partial_u^r\partial_v^rF_{\bm{x}}(u,v)
  \right|_{u=v=0}
  =
  \|\partial_{\bm v}^{\,r}
    \psi_{0,\Lambda}(\bm{x})\|_{\widehat H_0}^2
  \le C_rn^r,
\end{equation*}
proving Eq.~\eqref{eq:sm-ground-jet-energy}.
 
\end{proof}

We now apply the derivative bounds  to the sample patterns.  Define
\begin{equation*}
  \ket{\Phi_{d,\Lambda}}
  =\sum_{j=1}^Kc_j
  \ket{\psi_{0,\Lambda}(\bm{x}_0+d\bm{s}_j)},
  \qquad
  S=\max_j\norm{\bm{s}_j}.
\end{equation*}
Here the coefficients $c_j$ satisfy the moment conditions
in Eq.~\eqref{eq:sm-low-moment}, and $S$ is the maximum
norm of the sampling directions.
If some $\bm{s}_j=0$, the exact target ground state already belongs to the
sampled space and the claimed variational bound is immediate.  We henceforth
assume the stencil contains no target point.

First consider the near regime, there exits a $\rho$ such that:
\begin{equation*}
  nd^2S^2\le\rho .
\end{equation*}
One-dimensional Taylor expansion along the $j$th ray gives
\begin{equation*}
  |\psi_{0,\Lambda}(\bm{x}_0+d\bm{s}_j)\rangle
  =
  \sum_{r=0}^{m-1}\frac{d^r}{r!}
  \left.
    \partial_{\bm{s}_j}^{\,r}
    |\psi_{0,\Lambda}(\bm{x})\rangle
  \right|_{\bm{x}=\bm{x}_0}
  +
  d^m|R_{j,m,\Lambda}(d)\rangle,
\end{equation*}
with
\begin{equation*}
  |R_{j,m,\Lambda}(d)\rangle
  =
  \frac{1}{(m-1)!}
  \int_0^1(1-t)^{m-1}
  \left.
    \partial_{\bm{s}_j}^{\,m}
    |\psi_{0,\Lambda}(\bm{x})\rangle
  \right|_{\bm{x}=\bm{x}_0+td\bm{s}_j}
  \,dt.
\end{equation*}
Eq.~\eqref{eq:sm-ground-directional-polynomial} and the moment conditions
in Eq.~\eqref{eq:sm-low-moment} cancel the terms of degrees $1,\ldots,m-1$.
Consequently,
\begin{equation}
  \ket{\Phi_{d,\Lambda}}
  =\ket{\psi_{0,\Lambda}(\bm{x}_0)}
  +d^m\ket{R_{m,\Lambda}(d)},
  \qquad
  \ket{R_{m,\Lambda}(d)}
  =\sum_jc_j\ket{R_{j,m,\Lambda}(d)}.
  \label{eq:sm-ground-Phi-remainder}
\end{equation}

Both the Hilbert norm and the positive energy seminorm in
Eq.~\eqref{eq:sm-ground-energy-seminorm} satisfy the triangle inequality.
Applying Lemma~\ref{lem:sm-ground-jet-bounds} pointwise inside the integrals,
rather than forming a matrix element between different parameter points,
gives
\begin{equation*}
  \norm{R_{m,\Lambda}(d)}\le C_mn^{m/2},
  \qquad
  \norm{R_{m,\Lambda}(d)}_{\widehat H_0}\le C_mn^{m/2}.
\end{equation*}
It also gives the denominator estimate
\begin{equation*}
  \norm{\Phi_{d,\Lambda}}
  \ge1-C_m(nd^2)^{m/2}.
\end{equation*}
Since the sampling directions and coefficients are fixed, $\rho$ can be chosen so that the right-hand side
is at least $1/2$ throughout the near regime.  Using
$\widehat H_{0,\Lambda}\ket{\psi_{0,\Lambda}(\bm{x}_0)}=0$ and the variational
principle, we obtain
\begin{align}
  0
   & \le \widetilde E_{0,\Lambda}(d)-E_{0,\Lambda}(\bm{x}_0)\notag \\
   & \le
  \frac{d^{2m}
    \norm{R_{m,\Lambda}(d)}_{\widehat H_0}^2}
  {\norm{\Phi_{d,\Lambda}}^2}
  \le C_m n^m d^{2m}.
  \label{eq:sm-ground-near-bound}
\end{align}

It remains to treat the complementary regime $nd^2S^2>\rho$.  Choose any
fixed sampled direction $\bm{s}_{j_*}$.  The corresponding normalized sampled
ground state is an admissible Ritz vector, so
Eq.~\eqref{eq:sm-ground-single-snapshot} gives
\begin{equation}
  0\le
  \frac{\widetilde E_{0,\Lambda}(d)-E_{0,\Lambda}(\bm{x}_0)}{n}
  \le C_S d^2 .
  \label{eq:sm-ground-far-snapshot}
\end{equation}
Since $nd^2S^2>\rho$, we have
\begin{equation*}
  d^2
  \le\left(\frac{S^2}{\rho}\right)^q
  n^q d^{2(q+1)}.
\end{equation*}
Combining this inequality with Eq.~\eqref{eq:sm-ground-far-snapshot}, and
dividing Eq.~\eqref{eq:sm-ground-near-bound} by $n$ in the near regime, yields
in both cases
\begin{equation}
  0\le
  \frac{\widetilde E_{0,\Lambda}(d)-E_{0,\Lambda}(\bm{x}_0)}{n}
  \le A_q n^q d^{2(q+1)},
\label{eq:sc-fixed-order-input}
\end{equation}
where $A_q$ is independent of $n$ and $d$.  This proves the ground-state
prefactor estimate in Thm.~1.

\section{Sample complexity from local analytic response}
\label{sec:sample-complexity-linked-analytic}

The proof of Thm.~1 has already made the system-size
dependence explicit through the factor $n^q$.
The assumption below controls how the remaining
prefactor $A_q$ grows with the response order $q$.
\begin{assumption}
\label{ass:sc-linked-analytic-response}
Let $V$ be the fixed parameter neighborhood of
$\bm{x}_0$ used above, independent of the system size.
For a unit direction $\bm v$ and a base point
$\bm{x}=\bm{x}_0+t\bm v$ on a radial path within $V$,
consider the overlap and energy kernels
$G_{\bm{x}}(u,v)$ and $F_{\bm{x}}(u,v)$ defined in
Eqs.~\eqref{eq:sm-ground-overlap-kernel}
and~\eqref{eq:sm-ground-energy-kernel}.
Here $u$ and $v$ are scalar displacements along
the direction $\bm v$.
Define
\begin{equation*}
  L_{\bm{x}}(u,v)=\log G_{\bm{x}}(u,v),
  \qquad
  J_{\bm{x}}(u,v)
  =\frac{F_{\bm{x}}(u,v)}{G_{\bm{x}}(u,v)},
\end{equation*}
where the logarithm is chosen continuously near
$u=v=0$ so that $L_{\bm{x}}(0,0)=0$.

We assume that these functions have convergent
Taylor expansions near $u=v=0$.
After removing one factor of the system size
$n=|\Lambda|$, their connected response coefficients
satisfy
\begin{align}
  \frac{1}{n}
  \left|[u^av^b]L_{\bm{x}}\right|
  &\le
  M_{\rm an}R_{\rm an}^{-(a+b)},
  && a+b\ge2,
  \label{eq:sc-analytic-overlap-coefficients}\\
  \frac{1}{n}
  \left|[u^av^b]J_{\bm{x}}\right|
  &\le
  M_{\rm an}R_{\rm an}^{-(a+b)},
  && a+b\ge1.
  \label{eq:sc-analytic-energy-coefficients}
\end{align}
Here $a,b$ are nonnegative integers, and
$[u^av^b]$ denotes the coefficient of $u^av^b$
in the Taylor expansion about $u=v=0$.
The positive constants $M_{\rm an}$ and $R_{\rm an}$
are independent of $\Lambda$, $a$, and $b$, and the
bounds hold uniformly over the indicated base points
and unit directions.
Physically, this assumption controls the growth of
connected responses as the response order increases.
The constant term $J_{\bm{x}}(0,0)$ is controlled
separately by the single-sampled-state energy bound.

For each response order $q=m-1$, with $m\ge1$,
choose $K_m$ sampling directions
$\bm{s}_j^{(m)}$ and coefficients $c_j^{(m)}$,
$j=1,\ldots,K_m$, satisfying the moment conditions
in Eq.~\eqref{eq:sm-low-moment} through order $m-1$.
We assume
\begin{equation}
  \sum_{j=1}^{K_m}
  |c_j^{(m)}|
  \left(1+\|\bm{s}_j^{(m)}\|\right)^m
  \le C_{\rm st}^{\,m},
  \label{eq:sc-stable-stencil}
\end{equation}
and
\begin{equation*}
  S_*=
  \sup_{m\ge1}\max_{1\le j\le K_m}
  \|\bm{s}_j^{(m)}\|<\infty.
\end{equation*}
The constants $C_{\rm st}$ and $S_*$ are independent
of the system size, the sampling scale $d$, and $m$.
These conditions control the growth of the sampling
coefficients and keep the sampling directions bounded.
The sampling scale $d$ is restricted so that every
segment
$\bm{x}_0+td\bm{s}_j^{(m)}$, $0\le t\le1$,
remains inside $V$.
\end{assumption}

The analytic window and the bounds in
Asm.~\ref{ass:sc-linked-analytic-response} are uniform in the volume and in
the expansion order. We may therefore fix $\rho_0>0$, independent of
$\Lambda$, $d$, and $m$, such that all linked estimates below hold whenever
\begin{equation*}
  n\norm{\bm{x}-\bm{x}_0}^2\le\rho_0
\end{equation*}
and the corresponding interpolation segments stay inside that window.
\begin{lemma}
\label{lem:sc-linked-to-remainder}
Let $m=q+1\ge1$. At this fixed order, write
$K=K_m$, $c_j=c_j^{(m)}$, and
$\bm{s}_j=\bm{s}_j^{(m)}$ for the sampling parameters
in Asm.~\ref{ass:sc-linked-analytic-response}.

The remainder in
Eq.~\eqref{eq:sm-ground-Phi-remainder} is
\begin{equation}
  |R_{m,\Lambda}(d)\rangle
  =
  \sum_{j=1}^{K}
  \frac{c_j}{(m-1)!}
  \int_0^1
  (1-t)^{m-1}
  \left.
    \partial_{\bm{s}_j}^{\,m}
    |\psi_{0,\Lambda}(\bm{x})\rangle
  \right|_{\bm{x}=\bm{x}_0+td\bm{s}_j}
  \,dt.
  \label{eq:sc-order-m-remainder}
\end{equation}

Suppose that
\begin{equation*}
  nd^2S_*^2\le\rho_0
\end{equation*}
and that every sampling parameter
$\bm{x}_0+td\bm{s}_j$, $0\le t\le1$,
remains inside $V$.
Under Asm.~\ref{ass:sc-linked-analytic-response},
there exist positive constants $B_\psi$ and $B_E$
such that
\begin{align}
  \|R_{m,\Lambda}(d)\|
  &\le B_\psi^m n^{m/2},
  \label{eq:sc-order-uniform-remainder-norm}\\
  \|R_{m,\Lambda}(d)\|_{\widehat H_0}
  &\le B_E^m n^{m/2}.
  \label{eq:sc-order-uniform-remainder-energy}
\end{align}
The constants may depend on $\rho_0$ and the fixed
response and sampling bounds, but are independent
of $\Lambda$, $d$, and $m$.
\end{lemma}
\begin{proof}
We first estimate derivatives along a unit direction
$\bm v$, at a base point $\bm{x}$ on the corresponding
radial path within $V$.
By Asm.~\ref{ass:sc-linked-analytic-response},
\begin{align*}
  \left|[u^av^b]L_{\bm{x}}\right|
  &\le n C^{a+b},
  && a+b\ge2,\\
  \left|[u^av^b]J_{\bm{x}}\right|
  &\le n C^{a+b},
  && a+b\ge1,
\end{align*}
for a constant $C$ independent of the volume,
the base point, the unit direction, and the
response orders.

In the parallel-transport gauge,
$L_{\bm{x}}(0,0)=0$ and the two degree-one
coefficients of $L_{\bm{x}}$ vanish.
Thus every factor of $L_{\bm{x}}$ in
\begin{equation*}
  G_{\bm{x}}
  =e^{L_{\bm{x}}}
  =1+L_{\bm{x}}+\frac{L_{\bm{x}}^2}{2!}+\cdots
\end{equation*}
contributes at least two powers of $u,v$ and
one factor of $n$.
The number of ways to distribute the powers of
$u,v$ among these factors grows at most exponentially
in the total degree; the factorial denominators
in the exponential expansion only reduce the bound.
Consequently, for $a+b\ge1$,
\begin{equation}
  \left|[u^av^b]G_{\bm{x}}\right|
  \le
  \begin{cases}
    C_G^{a+b}n^{(a+b)/2},
      & a+b \text{ even},\\
    C_G^{a+b}n^{(a+b-1)/2},
      & a+b \text{ odd},
  \end{cases}
  \label{eq:sc-overlap-coefficient-bound}
\end{equation}
where $C_G$ is independent of the volume and
the response order.
The constant coefficient is $G_{\bm{x}}(0,0)=1$.

Taking $a=b=m$ and using
\begin{equation*}
  \left.
    \partial_u^m\partial_v^mG_{\bm{x}}(u,v)
  \right|_{u=v=0}
  =
  (m!)^2[u^mv^m]G_{\bm{x}}
  =
  \|\partial_{\bm v}^{\,m}
    \psi_{0,\Lambda}(\bm{x})\|^2
\end{equation*}
gives the required Hilbert-space derivative bound
for unit directions.

We next estimate the energy-weighted derivatives
using
\begin{equation*}
  F_{\bm{x}}=G_{\bm{x}}J_{\bm{x}}.
\end{equation*}
Suppose that
$n\|\bm{x}-\bm{x}_0\|^2\le\rho_0$.
The single-sampled-state energy bound gives
\begin{equation*}
  |J_{\bm{x}}(0,0)|
  =
  \langle\widehat H_{0,\Lambda}\rangle_{\bm{x}}
  \le
  Cn\|\bm{x}-\bm{x}_0\|^2
  \le C\rho_0.
\end{equation*}
Its contribution to $[u^mv^m]F_{\bm{x}}$ is therefore
bounded by
\begin{equation*}
  |J_{\bm{x}}(0,0)|
  \left|[u^mv^m]G_{\bm{x}}\right|
  \le C\rho_0 C_G^{2m}n^m.
\end{equation*}

Every nonconstant coefficient of $J_{\bm{x}}$
contributes one factor of $n$.
The coefficient of $G_{\bm{x}}$ multiplying it
has total degree at most $2m-1$.
By Eq.~\eqref{eq:sc-overlap-coefficient-bound},
that coefficient contributes at most $n^{m-1}$,
apart from a factor exponential in its degree.
Thus each such product is bounded by
$C^{2m}n^m$, after enlarging $C$.
There are at most $(m+1)^2$ coefficient products
in $[u^mv^m](G_{\bm{x}}J_{\bm{x}})$.
Since $(m+1)^2\le4^m$ for $m\ge1$, their sum
can also be absorbed into an exponential bound:
\begin{equation*}
  \left|[u^mv^m]F_{\bm{x}}\right|
  \le C_F^{2m}n^m.
\end{equation*}
Here $C_F$ may depend on the fixed constant $\rho_0$,
but not on the volume or $m$.

Using
\begin{equation*}
  \left.
    \partial_u^m\partial_v^mF_{\bm{x}}(u,v)
  \right|_{u=v=0}
  =
  (m!)^2[u^mv^m]F_{\bm{x}}
  =
  \|\partial_{\bm v}^{\,m}
    \psi_{0,\Lambda}(\bm{x})\|_{\widehat H_0}^2
\end{equation*}
gives the energy-weighted derivative bound
for unit directions.
For a general direction, an $m$th directional
derivative acquires a factor $\|\bm v\|^m$.
We therefore obtain
\begin{align}
  \|\partial_{\bm v}^{\,m}
    \psi_{0,\Lambda}(\bm{x})\|
  &\le
  m! B_{\psi,0}^m n^{m/2}\|\bm v\|^m,
  \label{eq:sc-order-uniform-jet-norm}\\
  \|\partial_{\bm v}^{\,m}
    \psi_{0,\Lambda}(\bm{x})\|_{\widehat H_0}
  &\le
  m! B_{E,0}^m n^{m/2}\|\bm v\|^m,
  \label{eq:sc-order-uniform-jet-energy}
\end{align}
with constants independent of the volume and $m$.
The second bound applies when
$n\|\bm{x}-\bm{x}_0\|^2\le\rho_0$.

Along each sampling path, the lemma's assumptions give
\begin{equation*}
  n\|td\bm{s}_j\|^2
  \le nd^2S_*^2
  \le\rho_0,
  \qquad 0\le t\le1.
\end{equation*}
Thus both derivative bounds apply at every point
in the integral remainder.
Since
\begin{equation*}
  \frac{1}{(m-1)!}
  \int_0^1(1-t)^{m-1}\,dt
  =\frac{1}{m!},
\end{equation*}
the factorial in the derivative bounds cancels
the factorial from the Taylor remainder.
The triangle inequality then gives
\begin{align*}
  \|R_{m,\Lambda}(d)\|
  &\le
  B_{\psi,0}^m n^{m/2}
  \sum_{j=1}^{K}|c_j|\|\bm{s}_j\|^m,\\
  \|R_{m,\Lambda}(d)\|_{\widehat H_0}
  &\le
  B_{E,0}^m n^{m/2}
  \sum_{j=1}^{K}|c_j|\|\bm{s}_j\|^m.
\end{align*}
The sampling-coefficient bound
in Eq.~\eqref{eq:sc-stable-stencil} implies
\begin{equation*}
  \sum_{j=1}^{K}|c_j|\|\bm{s}_j\|^m
  \le C_{\rm st}^{\,m}.
\end{equation*}
Setting
$B_\psi=B_{\psi,0}C_{\rm st}$ and
$B_E=B_{E,0}C_{\rm st}$
proves
Eqs.~\eqref{eq:sc-order-uniform-remainder-norm}
and~\eqref{eq:sc-order-uniform-remainder-energy}.
\end{proof}

For each $m\ge1$, recall the sampled-state combination
\begin{equation*}
  |\Phi_{d,\Lambda}\rangle
  =
  \sum_{j=1}^{K_m}c_j^{(m)}
  |\psi_{0,\Lambda}
  (\bm{x}_0+d\bm{s}_j^{(m)})\rangle.
\end{equation*}
Its dependence on $m$ is suppressed in the notation.
For a nontrivial sampled pattern $S_*>0$; if $S_*=0$, every sampled point is the
target and the variational error vanishes. Fix
\begin{equation*}
  \rho
  =
  \min\left\{
  \rho_0,\frac{S_*^2}{4B_\psi^2}
  \right\}.
\end{equation*}
Then $\rho$ is independent of $\Lambda$, $d$, and $m$. In the near regime
$nd^2S_*^2\le\rho$, Eqs.~\eqref{eq:sm-ground-Phi-remainder} and
\eqref{eq:sc-order-uniform-remainder-norm} give, for every $m\ge1$,
\begin{align*}
  \norm{\Phi_{d,\Lambda}}
   & \ge
  1-d^m\norm{R_{m,\Lambda}(d)}      \\
   & \ge
  1-\left(B_\psi^2nd^2\right)^{m/2} \\
   & \ge
  1-\left(\frac{\rho B_\psi^2}{S_*^2}\right)^{m/2}
  \ge
  1-2^{-m}
  \ge
  \frac12 .
\end{align*}
Thus the denominator bound is a consequence of the order-uniform remainder
estimate rather than an additional stencil assumption.

\begin{theorem}
\label{thm:sc-sample-complexity}
Under the hypotheses of the ground-state prefactor
result and Asm.~\ref{ass:sc-linked-analytic-response},
there exist positive constants $A$ and $C$,
independent of $n$, $d$, and the response order $q$,
such that the constants in
Eq.~\eqref{eq:sc-fixed-order-input} can be chosen
to satisfy
\begin{equation}
  A_q\le AC^q.
  \label{eq:sc-Aq-exponential}
\end{equation}
Consequently, for each nonnegative integer $q$,
\begin{equation}
  0\le\varepsilon_0
  :=
  \frac{\widetilde E_0(d)-E_0(\bm{x}_0)}{n}
  \le
  Ad^2(Cnd^2)^q.
  \label{eq:sc-renormalized-bound}
\end{equation}
Here $\widetilde E_0(d)$ is the lowest Ritz energy
in the corresponding sampled variational space.
Its dependence on the response order and sampling
scheme is suppressed in the notation.
We restrict to $0<d\le d_0$, where $d_0>0$ is
independent of the volume and response order and
is chosen so that all sampling paths remain inside
$V$ and the single-sampled-state energy bound applies.

For $D\ge1$ parameter dimensions and $K\ge1$
sampling points, define
\begin{equation*}
  q_K
  =
  \max\left\{
    q\in\mathbb Z_{\ge0}:
    \binom{D+q}{q}\le K
  \right\},
\end{equation*}
where $\mathbb Z_{\ge0}$ denotes the nonnegative
integers.
For each $K$ under consideration, assume that the
sampling matrix
\begin{equation*}
  \mathsf V_{\bm\alpha,j}
  =
  \bm{s}_j^{\bm\alpha},
  \qquad
  |\bm\alpha|\le q_K,\quad j=1,\ldots,K,
\end{equation*}
has full row rank, and that the sampling coefficients
satisfy the moment conditions through order $q_K$
and the uniform sampling bounds in
Asm.~\ref{ass:sc-linked-analytic-response}.

If $Cnd^2<1$, then
\begin{equation}
  \varepsilon_0
  \le
  Ad^2
  \exp\!\left[
    -q_K|\log(Cnd^2)|
  \right].
  \label{eq:sc-exact-order-bound}
\end{equation}
For every $K\ge D+1$, this implies
\begin{equation}
  \varepsilon_0
  \le
  Ad^2
  \exp\!\left[
    -c_D|\log(Cnd^2)|K^{1/D}
  \right],
  \label{eq:sc-stretched-exponential}
\end{equation}
where $c_D>0$ depends only on $D$.

For any target energy-density accuracy $\varepsilon>0$,
a sufficient condition for
$\varepsilon_0\le\varepsilon$ is
\begin{equation}
  K
  \ge
  C_D
  \left[
    1+
    \max\left\{
      0,
      \frac{\log(Ad^2/\varepsilon)}
           {|\log(Cnd^2)|}
    \right\}
  \right]^D,
  \label{eq:sc-sample-complexity}
\end{equation}
provided the sampling scheme satisfies the conditions
above. The constant $C_D>0$ depends only on $D$.
\end{theorem}

\begin{proof}
If $S_*=0$, every sampled parameter is the target point
and the variational error vanishes.
We therefore assume $S_*>0$ and use the fixed
constant $\rho$ chosen above.

Let $m=q+1$.
First suppose that $nd^2S_*^2\le\rho$.
By Eq.~\eqref{eq:sm-ground-Phi-remainder},
\begin{equation*}
  |\Phi_{d,\Lambda}\rangle
  =
  |\psi_{0,\Lambda}(\bm{x}_0)\rangle
  +
  d^m|R_{m,\Lambda}(d)\rangle.
\end{equation*}
Since
\begin{equation*}
  \widehat H_{0,\Lambda}
  |\psi_{0,\Lambda}(\bm{x}_0)\rangle
  =0,
\end{equation*}
the variational principle, the bound
$\|\Phi_{d,\Lambda}\|\ge1/2$, and
Lemma~\ref{lem:sc-linked-to-remainder} give
\begin{align}
  \varepsilon_0
  &\le
  \frac{d^{2m}
    \|R_{m,\Lambda}(d)\|_{\widehat H_0}^2}
       {n\|\Phi_{d,\Lambda}\|^2}
  \notag\\
  &\le
  4B_E^{2m}n^{m-1}d^{2m}.
  \label{eq:sc-remainder-to-energy}
\end{align}
Using $m=q+1$, we obtain
\begin{equation*}
  \varepsilon_0
  \le
  4B_E^2d^2(B_E^2nd^2)^q.
\end{equation*}

Now suppose that $nd^2S_*^2>\rho$.
Any normalized sampled ground state is an admissible
trial state. The single-sampled-state energy bound
therefore gives
\begin{equation*}
  \varepsilon_0\le C_{\rm far}d^2,
\end{equation*}
where $C_{\rm far}$ is independent of the volume,
$d$, and $q$, because all sampling directions
have norm at most $S_*$.
Since $nd^2S_*^2>\rho$,
\begin{equation*}
  d^2
  \le
  \left(\frac{S_*^2}{\rho}\right)^q
  n^qd^{2(q+1)}
\end{equation*}
for every nonnegative integer $q$.
Hence
\begin{equation*}
  \varepsilon_0
  \le
  C_{\rm far}
  \left(\frac{S_*^2}{\rho}\right)^q
  n^qd^{2(q+1)}.
\end{equation*}

Combining the two regimes, we may choose
\begin{equation*}
  A=\max\{1,4B_E^2,C_{\rm far}\},
  \qquad
  C=\max\left\{
    1,B_E^2,\frac{S_*^2}{\rho}
  \right\}.
\end{equation*}
These constants are independent of $n$, $d$, and $q$,
and give
\begin{equation*}
  A_q\le AC^q,
  \qquad
  \varepsilon_0\le Ad^2(Cnd^2)^q.
\end{equation*}
This proves
Eqs.~\eqref{eq:sc-Aq-exponential}
and~\eqref{eq:sc-renormalized-bound}.

We next relate the response order to the number
of sampling points.
The number of monomials in $D$ variables with
total degree at most $q$ is
\begin{equation*}
  N_q=\binom{D+q}{q}.
\end{equation*}
Thus the sampling matrix has $N_q$ rows and $K$
columns. Full row rank requires $K\ge N_q$ and
ensures that the moment conditions have a solution.
The assumed sampling bounds additionally control
the size of the chosen coefficients.

For a sampling scheme satisfying these conditions
through order $q_K$, we can set $q=q_K$ in
Eq.~\eqref{eq:sc-renormalized-bound}.
When $Cnd^2<1$, this gives
\begin{equation*}
  \varepsilon_0
  \le
  Ad^2(Cnd^2)^{q_K}
  =
  Ad^2
  \exp\!\left[
    -q_K|\log(Cnd^2)|
  \right],
\end{equation*}
proving Eq.~\eqref{eq:sc-exact-order-bound}.

 For $K\ge D+1$, the definition of $q_K$ gives
$q_K\ge1$ and
\begin{equation*}
  K<\binom{D+q_K+1}{D}.
\end{equation*}
For every integer $q\ge1$,
\begin{equation*}
  \binom{D+q+1}{D}
  =
  \prod_{j=1}^{D}\frac{q+j+1}{j}
  \le
  \binom{D+2}{2}q^D.
\end{equation*}
Consequently,
\begin{equation*}
  q_K\ge c_DK^{1/D},
  \qquad
  c_D=\binom{D+2}{2}^{-1/D}.
\end{equation*}
This estimate holds for every $K\ge D+1$.
Substituting it into
Eq.~\eqref{eq:sc-exact-order-bound}
proves Eq.~\eqref{eq:sc-stretched-exponential}.

Finally, let $\varepsilon>0$, and choose the smallest
nonnegative integer $q$ satisfying
\begin{equation*}
  q\ge
  \frac{\log(Ad^2/\varepsilon)}
       {|\log(Cnd^2)|}.
\end{equation*}
Then $Ad^2(Cnd^2)^q\le\varepsilon$ and
\begin{equation*}
  q\le
  1+
  \max\left\{
    0,
    \frac{\log(Ad^2/\varepsilon)}
         {|\log(Cnd^2)|}
  \right\}.
\end{equation*}
Moreover,
\begin{align*}
  N_q
  &=\frac{(q+1)\cdots(q+D)}{D!}\\
  &\le\frac{(q+D)^D}{D!}\\
  &\le
  \frac{(D+1)^D}{D!}
  \left[
    1+
    \max\left\{
      0,
      \frac{\log(Ad^2/\varepsilon)}
           {|\log(Cnd^2)|}
    \right\}
  \right]^D.
\end{align*}
Thus Eq.~\eqref{eq:sc-sample-complexity} holds,
for example, with
\begin{equation*}
  C_D=\frac{(D+1)^D}{D!}.
\end{equation*}
Indeed, this condition ensures $K\ge N_q$ and
therefore $q_K\ge q$.
For an admissible sampling scheme,
\begin{equation*}
  \varepsilon_0
  \le Ad^2(Cnd^2)^{q_K}
  \le Ad^2(Cnd^2)^q
  \le\varepsilon.
\end{equation*}
This proves the theorem.
\end{proof}

\section{Dimerized chain numerical simulation}
\label{sec:sm-numerics}

\subsection{Model and benchmark setting}
\label{sec:sm-model}

We test the scaling laws discussed in the main text using a spin-$\frac{1}{2}$ Heisenberg chain with open boundary conditions.
The Hamiltonian is
\begin{equation*}
  H(\bm x)
  =
  \sum_{i=1}^{n-1}
  \mathcal J_i\,
  \mathbf S_i\cdot\mathbf S_{i+1},
  \qquad
  \mathbf S_i=(S_i^x,S_i^y,S_i^z).
\end{equation*}
The nearest-neighbor couplings have a period-four modulation,
\begin{equation*}
  \mathcal J_i=J_a,
  \qquad
  a=1+[(i-1)\bmod 4].
\end{equation*}
Throughout the simulations we fix
\begin{equation*}
  J_1=J_3=1
\end{equation*}
and use
\begin{equation*}
  \bm x=(J_2,J_4)
\end{equation*}
as a two-dimensional control parameter.
The target point used in the benchmark calculations is
\begin{equation*}
  \bm x_0=(0.5,0.4),
\end{equation*}
corresponding to the coupling pattern
\begin{equation*}
  (J_1,J_2,J_3,J_4)=(1,0.5,1,0.4).
\end{equation*}
The calculations are performed in the regime $J_2,J_4<1$.
In this regime the chain is in a dimerized valence-bond-crystal phase with a finite bulk excitation gap and exponentially decaying correlations.
It therefore provides a controlled setting for testing the local Taylor scaling of wavefunction subspaces.

\subsection{Sampling basis}
\label{sec:sm-sampling-basis}

For a given sampling distance $d$, the sampled parameter points are written as
\begin{equation*}
  \bm x_j=\bm x_0+d\,\bm s_j,
  \qquad
  j=1,\ldots,K,
\end{equation*}
where $\bm s_j$ are dimensionless sample directions.
At each $\bm x_j$, the reference many-body states are computed independently by DMRG.

For the ground-state benchmark, the sampled variational space is
\begin{equation*}
  \mathcal B_{1,K}(d)
  =
  \operatorname{span}
  \bigl\{
  \lvert\psi_0(\bm x_j)\rangle
  \bigr\}_{j=1}^{K}.
\end{equation*}
For the low-energy benchmark, we keep the two lowest singlet states at each sampling point and use the pooled space
\begin{equation*}
  \mathcal B_{2,K}(d)
  =
  \operatorname{span}
  \bigl\{
  \lvert\psi_0(\bm x_j)\rangle,
  \lvert\psi_1(\bm x_j)\rangle
  \bigr\}_{j=1}^{K}.
\end{equation*}
All states in the low-energy benchmark are computed in the total-spin $S=0$ sector.
This restriction removes the triplet sector from the projected problem and isolates the scaling of a low-energy singlet subspace.

\subsection{Generation of Taylor stencils}
\label{sec:sm-stencil-generation}

The sampled directions $\bm s_j$ are chosen so that the sampled states can reproduce the Taylor expansion of the target wavefunction subspace up to a prescribed degree.
For a $q$ response order in $D$ parameter dimensions, define the multi-index set
\begin{equation*}
  \mathcal A_q
  =
  \left\{
  \bm\alpha=(\alpha_1,\ldots,\alpha_D)\in\mathbb N_0^D:
  |\bm\alpha|\le q
  \right\}.
\end{equation*}
Its cardinality is
\begin{equation*}
  N_q=\binom{D+q}{q}.
\end{equation*}
The minimal stencil uses $K=N_q$ points.
For a candidate sampled directions $\{\bm s_j\}_{j=1}^{K}$, we form the multivariate Vandermonde matrix
\begin{equation*}
  V_{\bm\alpha j}
  =
  \bm s_j^{\bm\alpha}
  =
  \prod_{\mu=1}^{D}
  (s_{j,\mu})^{\alpha_\mu},
  \qquad
  \bm\alpha\in\mathcal A_q .
\end{equation*}
The corresponding moment coefficients $c_j$ satisfy
\begin{equation*}
  \sum_{j=1}^{K}c_j\bm s_j^{\bm\alpha}
  =
  \delta_{\bm\alpha,\bm 0},
  \qquad
  |\bm\alpha|\le q .
\end{equation*}
Equivalently,
\begin{equation*}
  Vc=e_{\bm 0},
\end{equation*}
where $e_{\bm 0}$ denotes the unit vector associated with the constant monomial.
These moment conditions cancel all Taylor terms up to degree $q$ and leave a wavefunction-subspace error of order $d^{q+1}$.

In the actual projected-subspace calculation, the final variational amplitudes are obtained from Ritz diagonalization.
The coefficients $c_j$ are not used as the final variational coefficients.
They are used to certify and diagnose that the sampling geometry has the desired Taylor resolution.

The sampled pattern search is performed over integer candidate points
\begin{equation*}
  \bm s\in[-R,R]^D\cap\mathbb Z^D,
\end{equation*}
with the origin excluded.
Excluding $\bm s=\bm 0$ ensures that the target state itself is not included among the sampled basis states.
For each candidate subset, we compute $V$ and require
\begin{equation*}
  \operatorname{rank}V=N_q .
\end{equation*}
Among the full-rank sampled patterns, we choose the one minimizing
\begin{equation*}
  \mathcal S
  =
  \kappa(V)
  \left(
  1+0.01\sum_{j=1}^{K}|c_j|
  \right),
\end{equation*}
where $\kappa(V)$ is the condition number of the Vandermonde matrix.
The condition number controls the numerical stability of the moment constraints.
The coefficient norm weakly penalizes stencils that rely on large cancellations.
When two sampled patterns have indistinguishable scores within the numerical tolerance, we choose the one with the smaller maximal radius
\begin{equation*}
  R_s=\max_j \|\bm s_j\|.
\end{equation*}
If this still does not resolve the tie, we choose the stencil with the smaller value of $\sum_j |c_j|$.

For comparisons between different response orders, we also use nested sampled patterns.
A $q'$ response order with $q'>q$ is constructed by keeping all points of the degree-$q$ sampled patterns and adding new integer-grid points until the total number reaches
\begin{equation*}
  N_{q'}=\binom{D+q'}{q'}.
\end{equation*}
The enlarged sampled pattern is again required to have full Vandermonde rank for all monomials with $|\bm\alpha|\le q'$.
The same conditioning criterion is then applied to select the added points.
This construction keeps the lower-order sampling geometry fixed while increasing the Taylor resolution.

For overcomplete tests with $K>N_q$, the Vandermonde matrix has size $N_q\times K$.
We require full row rank and choose the minimum-norm solution of the moment equations,
\begin{equation*}
  c
  =
  V^{T}
  (VV^{T})^{-1}
  e_{\bm 0}.
\end{equation*}
For every sampled pattern, we record the rank, singular values, condition number, moment residuals, coefficient norms, and maximal radius.

\subsection{DMRG reference states}
\label{sec:sm-dmrg-reference}

The reference states are computed using density matrix renormalization group calculations with explicit $\mathrm{SU}(2)$ spin-rotation symmetry.
The matrix product states are represented in symmetry-adapted form and are labeled by the total spin quantum number.
This implementation preserves total spin exactly within the numerical representation and reduces the effective bond dimension required to reach a fixed accuracy.

The ground state at each sampling point is obtained by standard variational DMRG in the $S=0$ sector.
The first excited singlet is obtained by a deflated DMRG calculation.
After the ground state $\lvert\psi_0\rangle$ has converged at a given parameter point, we minimize the shifted Hamiltonian
\begin{equation*}
  H_{\mathrm{eff}}
  =
  H
  +
  \lambda
  \lvert\psi_0\rangle\langle\psi_0\rvert
\end{equation*}
within the same $S=0$ sector.
We use
\begin{equation*}
  \lambda=5
\end{equation*}
for all system sizes and parameter values.
The penalty term shifts the previously obtained ground state upward in energy, so that the lowest state of $H_{\mathrm{eff}}$ gives the first singlet excitation of the original Hamiltonian $H$.
The resulting state is then used as $\lvert\psi_1\rangle$ in the low-energy projected-subspace calculation.

The DMRG bond dimension is increased adaptively during the sweeps.
For all reference states used in the benchmarks, the discarded weight is reduced below $10^{-16}$.
At this level, the remaining DMRG truncation error is much smaller than the projection errors studied below.
The dominant numerical errors are associated with the finite sampling distance and with the conditioning of the projected overlap matrix.

\subsection{Projected-subspace diagonalization}
\label{sec:sm-projected-diagonalization}

After the sampling basis has been constructed, the target Hamiltonian is projected onto the nonorthogonal subspace.
Let
\begin{equation*}
  \mathcal B
  =
  \operatorname{span}
  \bigl\{
  \lvert\phi_a\rangle
  \bigr\}_{a=1}^{N_B}
\end{equation*}
denote either $\mathcal B_{1,K}$ or $\mathcal B_{2,K}$.
Here $\lvert\phi_a\rangle$ denotes the pooled set of retained snapshot states.
At the target parameter $\bm x_0$, we form
\begin{equation*}
  S_{ab}
  =
  \langle\phi_a|\phi_b\rangle,
  \qquad
  H_{ab}
  =
  \langle\phi_a|H(\bm x_0)|\phi_b\rangle .
\end{equation*}
The variational energies are obtained from the generalized eigenvalue problem
\begin{equation*}
  Hc=ESc .
\end{equation*}

Since nearby reference states can be nearly linearly dependent, the overlap matrix is orthogonalized before solving the projected problem.
We diagonalize
\begin{equation*}
  S=U\Lambda U^\dagger,
\end{equation*}
where
\begin{equation*}
  \Lambda=\operatorname{diag}(\lambda_1,\ldots,\lambda_{N_B}).
\end{equation*}
Only eigenmodes satisfying
\begin{equation*}
  \lambda_\nu
  \ge
  \epsilon_{\mathrm{svd}}\lambda_{\max}
\end{equation*}
are retained.
We use $\epsilon_{\mathrm{svd}}=3\times 10^{-7}$.
Let $U_{\mathrm r}$ be the matrix of retained eigenvectors and let $\Lambda_{\mathrm r}$ be the corresponding diagonal matrix of retained eigenvalues.
The singular-value spectra of the projected overlap matrices are shown in Fig.~\ref{fig:svd-spectrum}.
For all system sizes and sampling distances used in the benchmark,
the retained singular directions remain well separated from the cutoff.
This confirms that the projected Ritz calculation is not controlled
by numerical null directions of the snapshot basis.

The orthonormalized projected Hamiltonian is
\begin{equation*}
  \widetilde H
  =
  \Lambda_{\mathrm r}^{-1/2}
  U_{\mathrm r}^\dagger
  H
  U_{\mathrm r}
  \Lambda_{\mathrm r}^{-1/2}.
\end{equation*}
We then solve the ordinary Ritz problem
\begin{equation*}
  \widetilde H\widetilde c
  =
  E\widetilde c .
\end{equation*}
The eigenvalues are the variational estimates of the low-energy spectrum at $\bm x_0$.
The corresponding coefficient vector in the original nonorthogonal basis is
\begin{equation*}
  c
  =
  U_{\mathrm r}
  \Lambda_{\mathrm r}^{-1/2}
  \widetilde c .
\end{equation*}
The overlap truncation removes numerically ill-conditioned directions while keeping the Ritz calculation variational within the retained subspace.
\begin{figure}
  \centering
  \includegraphics[width=\textwidth]{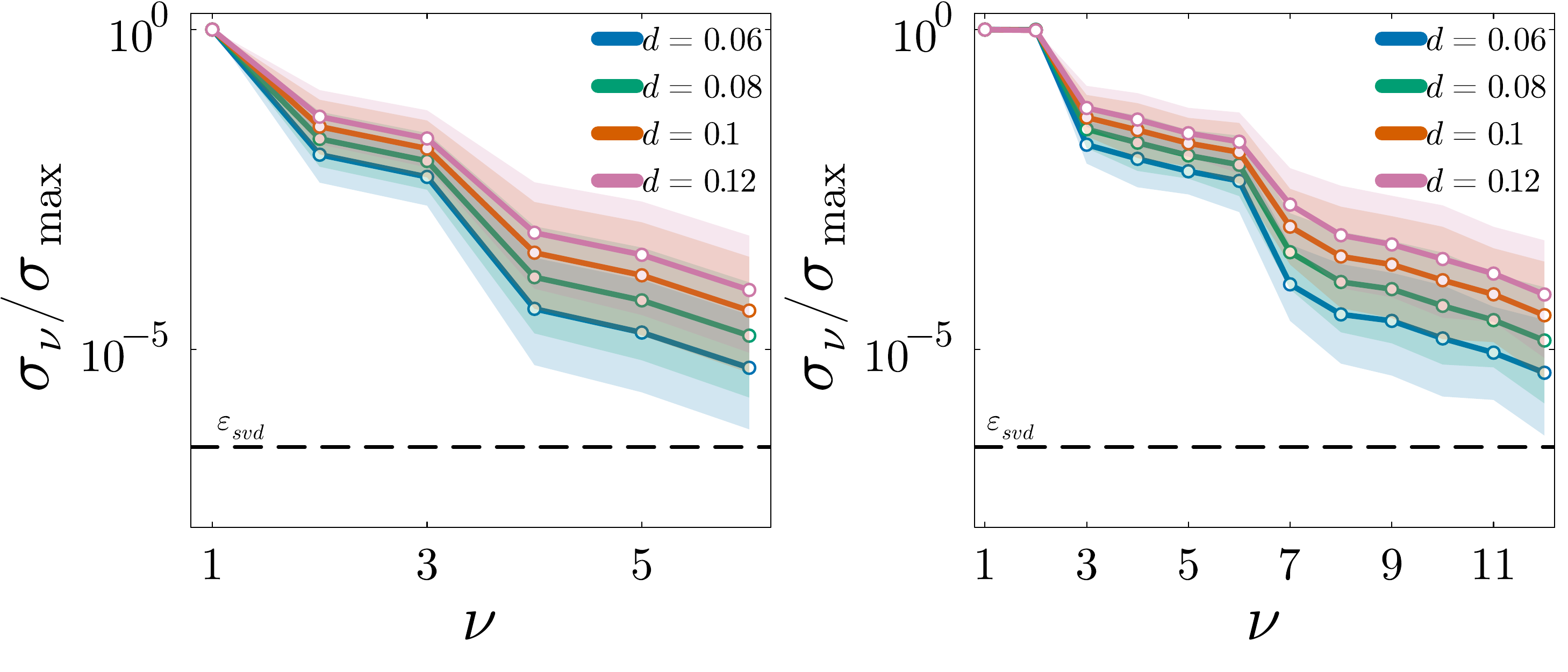}
  \caption{
    Singular-value spectra of the projected overlap matrix at fixed Taylor degree $q=1$.
    (a) Ground-state sampling basis.
    (b) Low-energy singlet sampling basis.
    For each sampling distance $d$, the solid line shows the median over the benchmark system sizes, and the shaded region indicates the full range over $n=36,64,128,256$.
    The dashed horizontal line marks the cutoff $\epsilon_{\mathrm{svd}}$ used in the overlap-matrix truncation.
  }
  \label{fig:svd-spectrum}
\end{figure}

\subsection{Comparison with random sampling}
\label{sec:sm-random-sampling}

To test the role of the sampling geometry, we compare the sampled patterns used in the main benchmark with a random sampling baseline.
For a fixed target point $\bm x_0$, sampling distance $d$, and Taylor degree $q$, the deterministic  pattern uses
\begin{equation*}
  \bm x_j=\bm x_0+d\,\bm s_j,
  \qquad
  j=1,\ldots,K,
\end{equation*}
where the dimensionless points $\bm s_j$ satisfy the Vandermonde moment conditions described in Sec.~\ref{sec:sm-stencil-generation}.
The number of sampling points is fixed to
\begin{equation*}
  K=N_q=\binom{D+q}{q}.
\end{equation*}

For the random baseline, we use the same number of sampling points and draw the dimensionless vectors $\bm s_j$ independently from an annulus,
\begin{equation*}
  r_1\le \|\bm s_j\|\le r_2 .
\end{equation*}
The corresponding physical sampling points are again $\bm x_j=\bm x_0+d\,\bm s_j$.
For each value of $d$, we generate 50 independent random batches.
For every random batch, we compute the sampled wavefunctions, construct the projected overlap and Hamiltonian matrices, apply the same overlap-matrix truncation, and solve the genelized eigenvalue problem.
The resulting distribution of energy-density errors is then compared with the error obtained from the conditioned Taylor stencil selected by the procedure in Sec.~\ref{sec:sm-stencil-generation}. The results are shown in Fig.~\ref{fig:random-sampling}.
\begin{figure}
  \centering
  \includegraphics[width=0.6\textwidth]{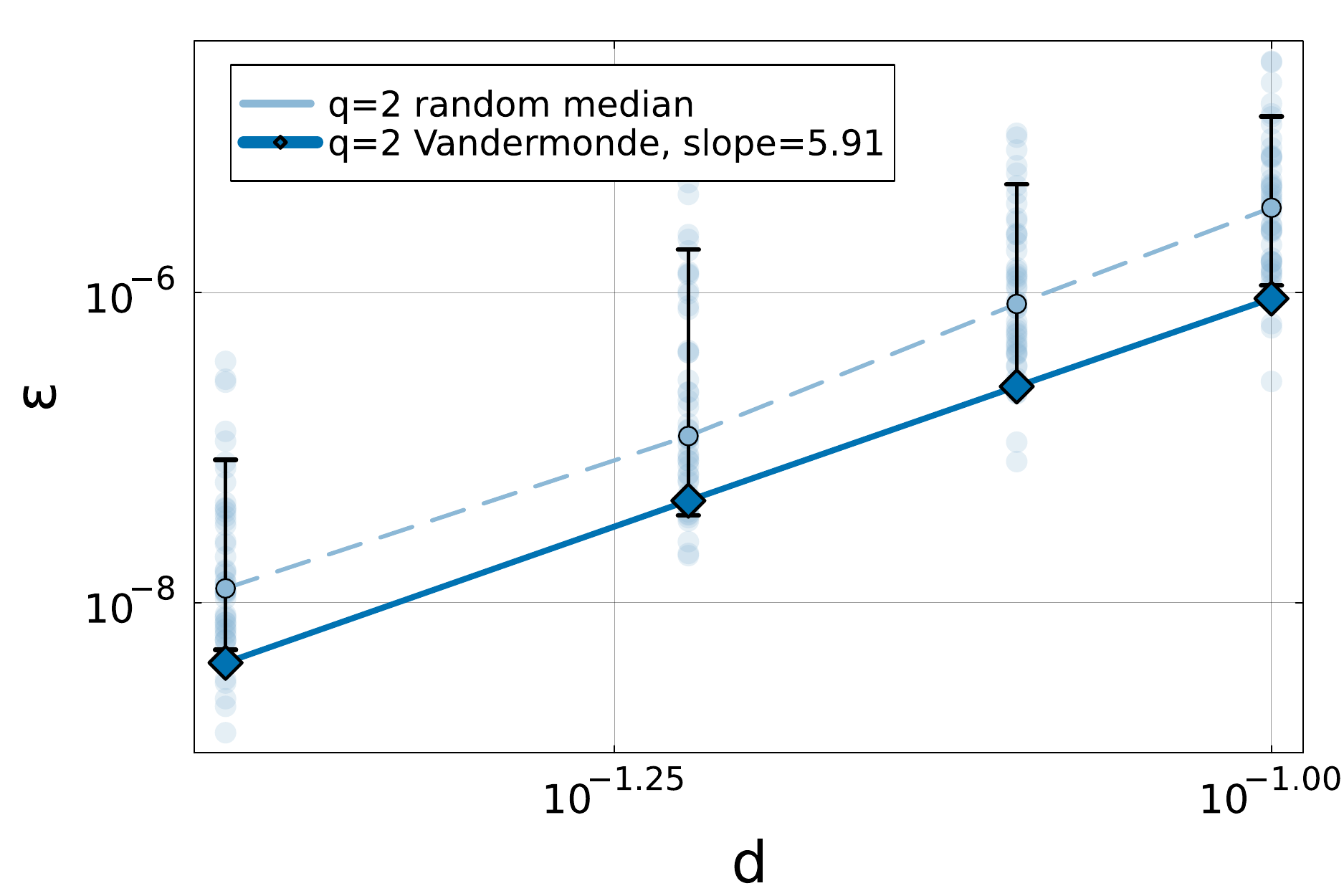}
  \caption{    Comparison between the conditioned Taylor stencil and random sampling
    at fixed sampled number.
    For each sampling distance $d$, the random baseline uses 50 independent
    batches of dimensionless sampling directions with the same number of
    points as the deterministic degree-$q$ stencil.
    The plotted errors are the resulting projected-subspace energy-density
    errors, computed with the same DMRG reference states, overlap truncation,
    and Ritz diagonalization protocol as in the deterministic calculation.
    The comparison isolates the effect of enforcing the Taylor moment
    conditions from the generic benefit of adding nearby sampled states.
  }
  \label{fig:random-sampling}
\end{figure}

\subsection{Error measures and diagnostics}
\label{sec:sm-error-diagnostics}

The projected-subspace energy error for the $a$th retained level is measured by the energy-density difference
\begin{equation}
  \varepsilon_a(d,n)
  =
  \frac{
    \widetilde E_a(d,n)-E_a(\bm x_0,n)
  }{n},
  \qquad
  a=0,1.
\end{equation}
Here $\widetilde E_a(d,n)$ is the Ritz energy obtained from the sampled basis at distance $d$, while $E_a(\bm x_0,n)$ is the reference DMRG energy at the target point.
For the ground-state benchmark, only $a=0$ is retained.
For the low-energy excited singlet benchmark, $a=0$ and $a=1$ denote the two lowest states in the $S=0$ sector.

For each projected calculation, we record the spectrum of the overlap matrix $S$, the number of retained directions after the cutoff, and the final Ritz energies.
For each stencil, we record the Vandermonde condition number, the rank, the moment residual, and the coefficient norms.
For each DMRG state, we record the discarded weight and convergence of the variational energy.
These diagnostics ensure that the observed scaling is controlled by the sampling construction rather than by DMRG truncation error or by an ill-conditioned projected basis.

As an additional check of the distance-scaling prediction beyond the
$q=1$ data shown in the main text, Fig.~\ref{fig:d-scaling-q0}
shows the corresponding results for $q=0$ and $q=2$.
The fitted slopes are close to the expected values $2(q+1)=2$
and $2(q+1)=6$, respectively, for both the ground-state and
low-energy excited singlet errors.
\begin{figure}[H]
  \centering
  \includegraphics[width=0.6\textwidth]{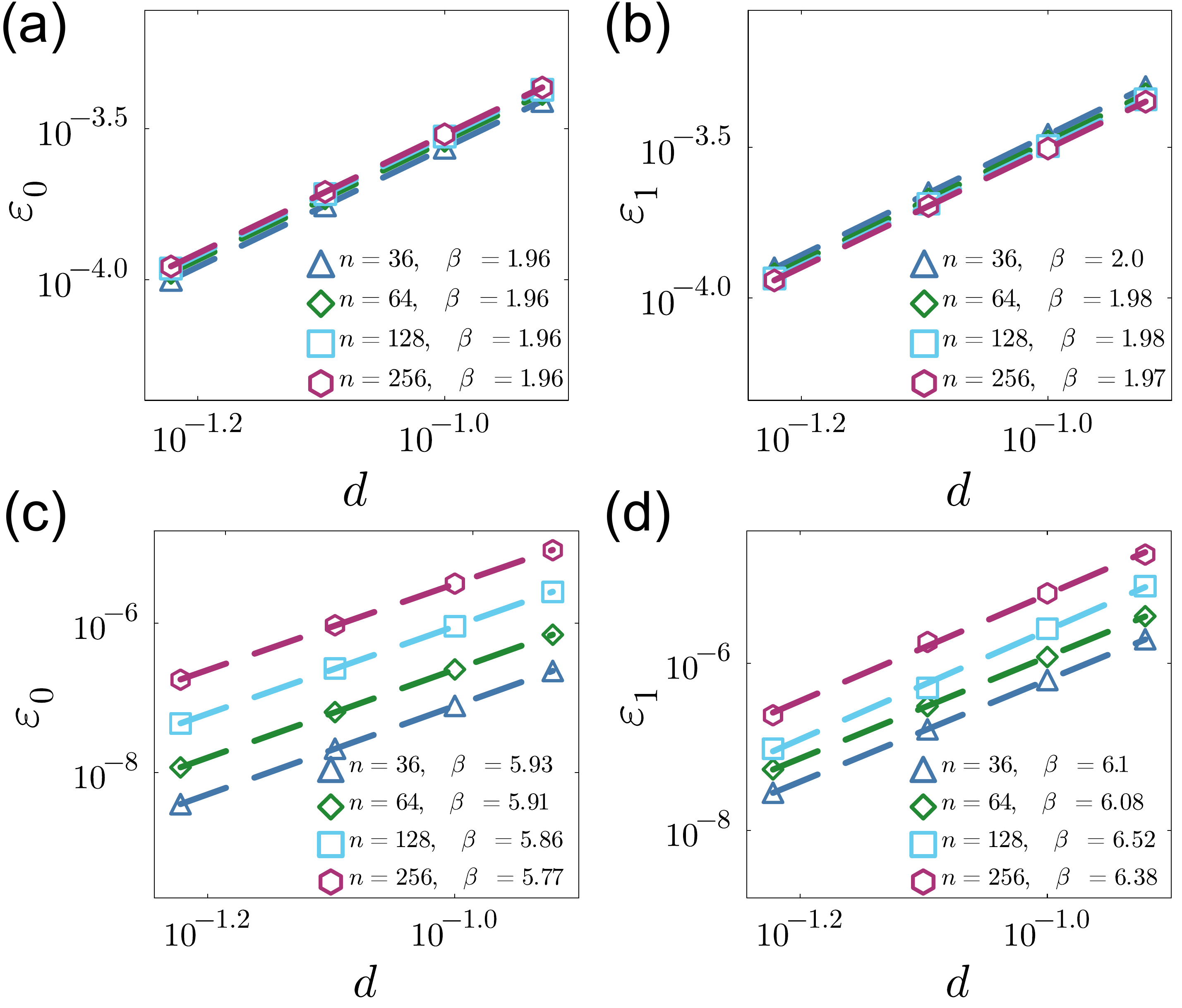}
  \caption{Ground state error $\varepsilon_0$ and lowest singlet error $\varepsilon_1$
    versus the sampling distance $d$ for several system sizes $n$. Dashed lines are power law fits.
    (a), (b) At order $q=0$, giving exponents close to the predicted value $2(q+1)=2$.
    (c), (d) At order $q=2$, giving exponents close to the predicted value $2(q+1)=6$.
  }
  \label{fig:d-scaling-q0}
\end{figure}

\section{Transverse-field Ising-chain numerical simulation}
\label{sec:sm-tfim}

\subsection{Model and retained low-energy sector}
\label{sec:sm-tfim-model}

To examine the crossover near a continuous quantum phase transition, we consider the transverse-field Ising chain with open boundary conditions,
\begin{equation*}
  H(g)
  =
  -\sum_{i=1}^{n-1}\sigma_i^z\sigma_{i+1}^z
  -g\sum_{i=1}^{n}\sigma_i^x .
\end{equation*}
The thermodynamic critical point is at
\begin{equation*}
  g_c=1.
\end{equation*}
For each finite chain, we retain the two lowest eigenstates,
\begin{equation*}
  \mathcal L(g)
  =
  \operatorname{span}
  \left\{
  \lvert\psi_0(g)\rangle,
  \lvert\psi_1(g)\rangle
  \right\},
\end{equation*}
and therefore set $M=2$.
On the ferromagnetic side, these states form the finite-size Ising doublet.
The external gap protecting the retained sector is
\begin{equation*}
  \Delta_{\mathrm{ext}}(g,n)
  =
  E_2(g,n)-E_1(g,n),
\end{equation*}
where $E_2$ is the third many-body level.
This choice treats the nearly degenerate doublet as a single spectral subspace and avoids following either member of the doublet separately.

At every sampling point, the two lowest states are computed by DMRG,
with the first excited state obtained using the same deflation procedure
as in the dimerized-chain calculation.
The convergence threshold, and
projected subspace diagonalization are the same as those used in
Secs.~\ref{sec:sm-dmrg-reference}
and~\ref{sec:sm-projected-diagonalization}.
No symmetry is imposed in the MPS calculations.
The target states at $g_0$ are used only to obtain benchmark energies
and are not included in the sampled variational basis.

\subsection{Symmetric one-dimensional sampled pattern}
\label{sec:sm-tfim-stencil}

The transverse field $g$ is the only control parameter, so the parameter-space dimension is $D=1$.
We use the degree-$q=1$ construction throughout the critical-crossover calculation.
The minimal sample then contains
\begin{equation*}
  K=N_q=\binom{D+q}{q}=2
\end{equation*}
sampling parameters.
Unlike the two-dimensional integer sampled pattern used for the dimerized chain, the one-dimensional pattern is chosen symmetrically about the target field,
\begin{equation*}
  s_-=-1,
  \qquad
  s_+=1,
\end{equation*}
which gives the physical sampled parameter
\begin{equation*}
  g_-=g_0-d,
  \qquad
  g_+=g_0+d.
\end{equation*}
For the monomial basis $\{1,s\}$, the Vandermonde matrix is
\begin{equation*}
  V
  =
  \begin{pmatrix}
    1  & 1 \\
    -1 & 1
  \end{pmatrix}.
\end{equation*}
It is nonsingular, and the moment equations are solved by
\begin{equation*}
  c_-=c_+=\frac{1}{2}.
\end{equation*}
Consequently,
\begin{equation*}
  c_-+c_+=1,
  \qquad
  c_-s_-+c_+s_+=0,
\end{equation*}
so the constant term is reproduced and the linear response order is  exactly.

The sampled low-energy trial space is therefore
\begin{equation*}
  \mathcal B_{2,2}(d;g_0)
  =
  \operatorname{span}
  \left\{
  \lvert\psi_a(g_0-d)\rangle,
  \lvert\psi_a(g_0+d)\rangle
  \,\middle|\,
  a=0,1
  \right\}.
\end{equation*}
Before overlap truncation, this space contains four snapshot states.
The overlap and projected Hamiltonian matrices are evaluated with the target Hamiltonian $H(g_0)$.
We use the same relative overlap cutoff,
\begin{equation*}
  \epsilon_{\mathrm{svd}}=3\times 10^{-7},
\end{equation*}
as in the dimerized-chain benchmark.
\subsection{Extraction of the effective convergence exponent}
\label{sec:sm-tfim-exponent}

For each target field $g_0$, system size $n$, and sampling distance $d$,
the projected calculation gives two Ritz energies
$\widetilde E_a(d,n;g_0)$ with $a=0,1$.
Their energy-density errors are
\begin{equation*}
  \varepsilon_a(d,n;g_0)
  =
  \frac{
    \widetilde E_a(d,n;g_0)-E_a(g_0,n)
  }{n}.
\end{equation*}
For each retained level, we extract an effective distance exponent from
the log-log fit
\begin{equation*}
  \log \varepsilon_a(d,n;g_0)
  =
  b_a(g_0,n)
  +
  \beta_a(g_0,n)\log d .
\end{equation*}
The same set of sampling distances and the same fitting protocol are used for every $g_0$ and $n$.

The exponent shown in Fig.~4 is extracted from the
ground-state energy-density error,
\begin{equation*}
  \beta(g_0,n)=\beta_0(g_0,n).
\end{equation*}
Representative raw data for $n=256$ are shown in
Fig.~\ref{fig:gs-error-vs-d}.
For the degree-$q=1$ stencil, the gapped prediction is
\begin{equation*}
  \varepsilon_0=O(d^4),
\end{equation*}
which generically corresponds to
\begin{equation*}
  \beta=2(q+1)=4.
\end{equation*}

The fitted quantity $\beta(g_0,n)$ is an effective exponent obtained
over a fixed finite range of sampling distances.
For every fixed finite $n$, the retained low-energy subspace remains
isolated and its spectral projector is analytic in a sufficiently small
neighborhood of $g_0$~\cite{kato1966perturbation}.
The asymptotic error therefore remains $O(d^4)$ as $d\to0$, and the
effective exponent generically approaches four.
Near $g_c$, however, the external gap decreases with increasing
system size, and the range of sampling distances over which the
asymptotic quartic behavior is observable becomes progressively
smaller.
\begin{figure}[]
  \centering
  \includegraphics[width=\textwidth]{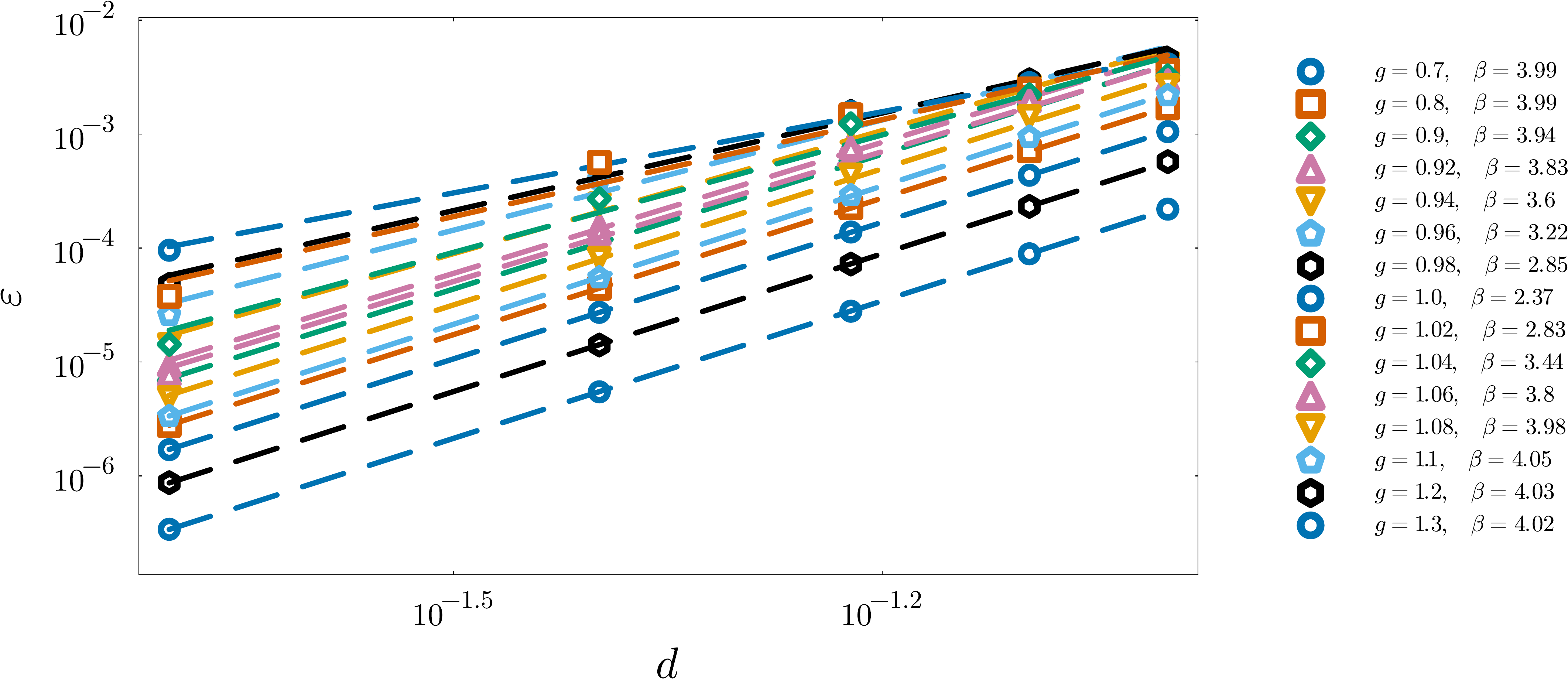}
  \caption{
    Distance dependence of the ground-state energy-density error
    $\varepsilon_0$ for the transverse-field Ising chain at different
    target fields $g$.
    The results are obtained for $n=256$ using the degree-$q=1$
    symmetric stencil and correspond to the $n=256$ data shown in
    Fig.~4.
    Symbols denote the numerical results, and all dashed lines show
    power-law fits $\varepsilon_0\propto d^\beta$ over the common
    fitting interval.
    The fitted effective exponents are listed in the legend.
    Away from the critical region, $\beta$ is close to the gapped
    prediction $\beta=4$.
    As $g$ approaches the critical point $g_c=1$, the error curves
    become progressively shallower and the fitted finite-window
    exponent decreases.
  }
  \label{fig:gs-error-vs-d}
\end{figure}
For the transverse-field Ising chain, the ground-state fidelity
susceptibility at criticality scales as $\chi_F\sim n^2$.
This scaling suggests the critical crossover variable
\begin{equation*}
  \chi_F d^2\sim (nd)^2,
\end{equation*}
in contrast to the gapped scaling variable $nd^2$.
The reduction of the finite-window exponent below four is therefore
consistent with a crossover out of the practically accessible gapped
Taylor regime.
It does not imply a change of the strict finite-size asymptotic order.

\end{document}